\documentclass[reqno]{amsart}
\usepackage{amssymb}
\usepackage[dvips]{epsfig}
\usepackage{graphicx}
\usepackage{color}
\usepackage{amsmath}
\usepackage{amsfonts}
\usepackage{amsthm}
\usepackage{bbm}
\usepackage{mathrsfs}

\newcommand{\R}{\mathbb R}

\newcommand{\Z}{\mathbb Z}

\newcommand{\C}{\mathbb C}
\newcommand{\Po}{\mathbb P}

\newcommand{\ep}{\varepsilon}

\newcommand{\g}{\gamma}
\newcommand{\et}{\eta}

\renewcommand{\lg}{\Lambda}

\renewcommand{\v}{\mathfrak{v}}
\newcommand{\N}{\mathbb{N}}

\newcommand{\ji}{\langle}
\newcommand{\jd}{\rangle}

\newcommand{\dist}{{\rm dist}}
\newcommand{\diam}{{\rm diam}}
\newcommand{\good}{{\rm good}}

\newtheorem{thm}{Theorem}[section]
\newtheorem{lem}[thm]{Lemma}

\theoremstyle{remark}
\newtheorem{rem}{\bf Remark}[section]
\theoremstyle{definition}
\newtheorem{defn}[thm]{Definition}

\numberwithin{equation}{section}

\usepackage[colorlinks=true,pdfstartview=FitV,linkcolor=magenta,citecolor=cyan]{hyperref}
\usepackage{bm}

\begin{document}

\title[Localization for  random  long-range model]{Localization for random operators on $\mathbb{Z}^d$ with the long-range hopping}
\author[Shi]{Yunfeng Shi}
\author[Wen]{Li Wen}
	\author[Yan]{Dongfeng Yan}
    \address[Shi]{School of Mathematics,
	 Sichuan University,
	 Chengdu 610064,
	China
	 }
	 \email{yunfengshi@scu.edu.cn}
	 	\address[Wen]{School of Mathematics,
	 Sichuan University,
	 Chengdu 610064,
	China
	 }
	\email{liwen.carol98@gmail.com}
\address[Yan] {School of Mathematics and Statistics,
	Zhengzhou University,
	Zhengzhou 450001,
	China}
\email{yandongfeng@zzu.edu.cn}

\date{}

\begin{abstract}
In this paper,  we investigate random operators on $\mathbb{Z}^d$ with  H\"older continuously distributed potentials   and the  long-range hopping. The hopping amplitude decays with the inter-particle distance  $\|\bm x\|$ as  $e^{-\log^{\rho}(\|\bm x\|+1)}$ with $\rho>1,\bm x\in\Z^d$. By employing the multi-scale analysis (MSA) technique, we prove  that  for  large disorder,   the random operators  have   pure point spectrum with localized  eigenfunctions  whose decay rate is the same  as  the hopping term.      This gives a  partial  answer to a conjecture  of  Yeung and Oono [{\it Europhys. Lett.} 4(9), (1987): 1061-1065].  
\end{abstract}

\maketitle
\section{Introduction}
The  random Schr\"odinger operator  (i.e., the Anderson model) on $\Z^d$ was first introduced  by Anderson  in the   seminal work  \cite{And58} to describe the motion of  noninteracting quantum particles in disordered media.   It turns out that the Anderson localization\footnote {In the present, by  the Anderson  localization we mean the (spectral) exponential localization, namely, pure point spectrum with exponentially decaying eigenfunctions.}  transitions properties for the Anderson model  rely  heavily on the dimension $d$, the
strength  of the disorder, and the energy. In the physical phase diagram,  it is believed that Anderson localization occurs for all energies and  all non-zero disorder if $d=1,2$, while for the  case of $d\geq 3$ and small disorder, there should exist the absolutely continuous spectrum in some  energy interval.  Mathematically, localization has been proven for three  regimes: 
(i) for all energies and arbitrary disorder in $d = 1$,  (ii) in any dimension and for all
energies at   large disorder, and (iii) near the edges of the spectrum in any dimension and for arbitrary disorder. Indeed, the first rigorous proof of the localization for random operators was
due to Goldsheid-Molchanov-Pastur \cite{GMP77}: they proved  the pure point spectrum for some one-dimensional  continuous random Schr\"odinger operators. A  proof of localization for the one dimensional Anderson model  was obtained by Kunz-Souillard \cite{KS80}. In higher dimensions, Fr\"ohlich-Spencer \cite{FS83} proved, either for  large disorder or extreme energies, the absence of diffusion for  the random Schr\"odinger operator  by developing the  remarkable  multi-scale analysis (MSA) method. Based on the Green's function estimates in  \cite{FS83},  \cite{FMSS85,DLS85,SW86} finally obtained the Anderson localization at either large disorder or extreme energies. An alternative method for the proof of the localization for random operators, known as the fractional moment method (FMM), was developed by Aizenman-Molchanov \cite{AM93}. This celebrated  method also has numerous applications in localization problems for random operators on $\Z^d$, such as  the first  proof of  both the dynamical localization  \cite{Aiz94}  and  power-law localization (for random operators $\Z^d$ with the power-law long-range hopping) \cite{AM93}.  However, the problem of proving the existence of the absolutely continuous spectrum for the Anderson model remains largely open.  For more results on the study of the localization for random operators, we refer to \cite{Kir08, AW15} and references therein.
 
While random models with finite-range hopping  (e.g., the Laplacian as in the Anderson model)  work nicely in describing variety of materials, the  long-range hopping is often found in different physical systems.  In mathematics,  the study of localization  type problems for long-range hopping operators with both random   potentials \cite{SS89, W91, Kle93, G93, AM93, JM99, Shi21,JS22, DMR24} and  quasi-periodic potentials \cite{JK16, JLS20, Shi22, Liu22,  SW22, Shi23, SW23, SW24} has attracted great attention over the years.   In this paper, we study  the following long-range random operator 
\begin{align}\label{y1}
    H_{\omega}=\varepsilon \Gamma_{\phi}+V_{\omega}(\bm x)\delta_{\bm x,\bm y}, \  \bm x,\bm y\in\Z^d,
\end{align}
where $\varepsilon>0$ represents the inverse of the disorder strength,  and  $\Gamma_{\phi}$ is a  translation invariant operator satisfying for $\forall \bm x,\bm y\in \mathbb{Z}^d$ and some $\g>0$, $\rho>1$, 
\begin{align}
    \label{y2}\Gamma_{\phi}(\bm x,\bm y)&=\phi(\bm x-\bm y)=\overline{\phi(\bm y-\bm x)},\\
    \label{y2-d} |\phi(\bm x)|&\le e^{-\gamma \log^{\rho}(\|\bm x\|+1)},\ \|\bm x\|=\max_{1\leq i\leq d}|x_i|.
\end{align}
For the  diagonal part,  we   let $\{V_{\omega}(\bm x)\}_{\bm x\in\Z^d}$ be a sequence of independent and identically distributed ({\it i.i.d})  random variables with a  common distribution $\mu$  on some  probability space $(\Omega,\mathcal{F},\mathbb{P})$ ($\mathcal{F}$ a $\sigma$-algebra on $\Omega$ and $\mathbb{P}$ a probability measure on $(\Omega,\mathcal{F})$).
We assume  that $\mu$ is H\"older continuous (cf. Definition \ref{holder} in the following for details). 
Note  that the operator $\Gamma_\phi$ satisfying  \eqref{y2-d} has previously been used in the construction of almost-periodic solutions for some nonlinear Hamiltonian equations \cite{Pos90}, and was  recently  introduced  by Shi-Wen \cite{SW22}  to the study of the localization for monotone quasi-periodic operators on $\Z^d$ via the KAM diagonalization approach. We also mention that the existence of  localized  eigenfunctions whose decay rate is the same as \eqref{y2-d}  has been established   in  \cite{ca17}.  The present work aims to prove via the MSA scheme that, for sufficiently small $\varepsilon$ and $\mathbb P$ a.e. $\omega$,   $H_\omega$ (given by \eqref{y1})  has  pure point spectrum with localized eigenfunctions whose decay rate is the same as \eqref{y2-d}.

The main motivations for investigating the long-range  model  \eqref{y1}  are twofold. The first one comes from  resolving  the {\bf conjecture}  of  \cite{YO87}:      {\it for $d=1$,  if the hopping term $|\phi(\bm x)|$  decays more slowly than any exponential function, but faster than $\frac{1}{\|\bm x\|}$, then such model  has pure point spectrum with  localized eigenstates whose decay rate is the same as the hopping term. Moreover, the result should be universal and  is  independent of particular configurations of the disorder}. This conjecture was tested \cite{YO87} based on numerical studies for the case $\phi(\bm x)=e^{-\sqrt{\|\bm x\|}}$. This paper tries to give an affirmative but partial answer to this conjecture in the more general  $d$-dimensional case from the mathematical perspective.  The second one lies in that we want to extend the remarkable MSA approach of Fr\"ohlich-Spencer \cite{FS83}   (cf. \cite{DK89} for the significant simplification of  this method) from the finite-range hopping to a slower long-range (e.g.,  the   \eqref{y2-d})   one.  In this respect,   Shi \cite{Shi21} has previously  introduced a novel MSA type approach to prove the spectral localization for random operators with power-law long-range hopping.

Our approach is mainly based on a MSA type Green's function estimates in the spirit of \cite{DK89, Kle93}.  However, the presence of the  weight \eqref{y2-d} leads to the technical challenge:  the function $f(\bm x, \bm y)=\log^{\rho}(\|\bm x-\bm y\|+1)$ ($\rho>1$) is {\it not a metric}.  Similar  issue also appears in the study of localization for  quasi-periodic operators with long-range hopping \eqref{y2} \cite{SW22}. Fortunately,  it was proved in \cite{SW22} that  the function  $f(\bm x, \bm y)$ is a {\it quasi-metric}  (cf. \eqref{qua} in the following  for details),  
which suffices for the  KAM diagonalization.  
In this paper,  we will prove  that  the {\it quasi-metric} estimate  is  also sufficient for performing a MSA scheme, which is one of  main contributions here.  As mentioned above, if $\rho=1$, then the long-range hopping \eqref{y2} becomes  a power-law one,  and certain  {\it tame}  estimate  is required in this type of  MSA \cite{Shi21}.  

If the distribution $\mu$ is absolutely continuous, both power-law localization and (sub)exponential localization can be established via the FMM  \cite{AM93}.  We believe the FMM can also deal with the hopping satisfying \eqref{y2-d} once $\mu$ is absolutely continuous.  Indeed, the Green's function estimates via the FMM  \cite{AM93,ASFH01} need  
a mild condition (such as the H\"older continuity) on the  distribution. However, proving the localization via FMM typically builds on the Simon-Wolff  criterion \cite{SW86}, which requires the distribution to be absolutely continuous.  In contrast, the MSA method  can be largely improved to treat completely singular distributions, such as the Bernoulli  ones  (cf. e.g., \cite{Bou04, BK05, GK13, DS20, LZ22} for  important works concerning Bernoulli type distributions).  In the present, we cannot handle the Bernoulli potentials, since our approach relies essentially on the priori Wegner estimate (this  depends on the H\"older continuity of the distribution).

The paper is organized as follows.  In \S \ref{not}, we introduce some useful notations.  The \S\ref{mthm} contains our main results on Green's function estimates (cf. Theorem \ref{thm1}) and the  localization (cf. Theorem \ref{thm2}). In \S\ref{tis}, we verify that Theorem \ref{thm1} holds for the initial step. In \S\ref{wg}, we  introduce the  Wegner  estimate.  In \S \ref{ITthm}, we present an iteration theorem. The proof of Theorem \ref{thm1}  is given in \S\ref{sVP2}--\S\ref{pot1}. In \S\ref{pot2}, we prove  Theorem \ref{thm2} via combining Theorem \ref{thm1} and the Shnol's  theorem. Some technical proofs are presented in  the Appendix. 

\section{The Notations}\label{not}
\begin{itemize}
	\item For $\bm n\in\R^d$, let
		$\|\bm n\|=\max_{1\le i\le d}|n_i|.$
	Denote by $\dist(\cdot,\cdot)$ the distance induced by  $\|\cdot\|$. Define for $\Lambda\subset\R^d,$
	\begin{align*}
		\diam(\lg)=\sup_{\bm k,\bm k'\in\lg}\|\bm k-\bm k'\|.
	\end{align*}
	\item For $\bm x\in\Z^d$ and $L>0$, define
	\begin{align}
		\label{y3} B_L(\bm x)&=\{\bm y\in\mathbb{Z}^d:\ \|\bm y-\bm x\|\le L\}.
	\end{align}
	Moreover, write $B_L=B_L(\bm 0)$.
	\item Let $\{\delta_{\bm x}\}_{\bm x\in\Z^d}$ denote  the standard basis of $\ell^2(\Z^d)$.
	\item Denote by $\ji\cdot,\cdot\jd$ the standard inner product on $\ell^2(\Z^d)$.
	\item Denote by  $\|\cdot\|_2$  the standard operator norm on $\ell^2(\mathbb{Z}^d)$.
	\item For $B\subset \Z^d$, denote by $R_B$  the standard  restriction operator.
	\item For $B\subset\Z^d$, define $H_B=R_B H_{\omega} R_B$. The spectrum of $H_B$ is denoted by $\sigma(H_B)$.  For $E\notin \sigma(H_B)$,  the Green's function  is defined by $\mathcal{G}_{B}^E=(H_B-E)^{-1}$.  Moreover, let 
	\begin{align}\label{y4}
		\mathcal{G}_{B}^E(\bm x,\bm y)=\ji\delta_{\bm x},(H_B-E)^{-1}\delta_{\bm y}\jd,\ \bm x,\bm y \in B.
	\end{align}
	\item $[x]$ denotes the integer part of $x\in\R$.
	\item $\# A$ denotes the cardinality of a finite set $A$.
\end{itemize}

\section{Main results}\label{mthm}
In this section, we  state  our main results of the present work. 

Recall that  our model $H_\omega$ is given by \eqref{y1} with  $\Gamma_\phi$
 satisfying \eqref{y2} and \eqref{y2-d}.


We first introduce some  useful definitions. 

\begin{defn}\label{defENR}
Let $E\in\mathbb{R}$ and {$1<\rho'<\rho<\rho'+1$}. We call a cube $B_L(\bm x)$ non-resonant with respect to $E$ ($E$-NR for short) if
\begin{align}\label{y5}
      \dist(E,\sigma(H_{B_L(\bm x)}))\ge e^{-\log ^{\rho'}L}.  
    \end{align} 
    Otherwise, we call $B_L(\bm x)$ resonant with respect to $E$ ($E$-R for short).
\end{defn}
    \begin{rem}
If $B_L(\bm x)$ is $E$-NR, we have
\begin{align}\label{ENR}
    \|\mathcal{G}_{B_L(\bm x)}^E\|_2=\frac{1}{\dist(E,\sigma(H_{B_L(\bm x)})}\le e^{\log^{\rho'}L}.
\end{align}
    \end{rem}

\begin{defn}\label{defkeg}
     Let $\kappa>0$. We say that a cube $B_L(\bm x)$ is $(\kappa,E)$-good, if it is $E$-NR and fulfills
    \begin{align}\label{y6}
        |\mathcal{G}_{B_L(\bm x)}^E(\bm x',\bm x'')|\le e^{-\kappa\log^{\rho}(\|\bm x'-\bm x''\|+1)}\ \text{for $\|\bm x'-\bm x''\|\ge L^{\frac{4}{5}}$}.
    \end{align}
   Otherwise, we say that $B_L(\bm x)$ is $(\kappa,E)$-bad. We say that $B_L(\bm x)$ is a $(\kappa,E)$-good (resp. $(\kappa,E)$-bad) $L$-cube if it is $(\kappa,E)$-good (resp. $(\kappa,E)$-bad).
\end{defn}

\begin{defn}[cf. \cite{CKM87}]\label{holder}
	We say that the distribution $\mu$ is H\"older continuous of order $\lambda>0$, provided that 
	\begin{align}\label{vv22}
		\frac{1}{\mathcal{D}_{\lambda}(\mu)}:=\inf\limits_{\beta>0}\sup\limits_{0<|b-a|\le\beta}\mu([a,b])|b-a|^{-\lambda}<\infty.
	\end{align}
	Denote by $\mathscr{H}(\lambda)$ the  set  of all distributions which are H\"older continuous of order $\lambda>0$.
\end{defn}
\begin{rem}
	If $\mu\in\mathscr{H}(\lambda)$, then for each $\beta\in(0,\mathcal{D}_{\lambda}(\mu))$, we can find some $\beta_0=\beta_0(\mu,\beta)>0$ such that
	\begin{equation}\label{vv23}
		\mu([a,b])\le \beta^{-1}|a-b|^{\lambda}\  {\rm for}\ 0\le b-a\le\beta_0.
	\end{equation}
	
\end{rem}

 Throughout this paper, we assume that $$1<\rho'<\rho<\rho'+1,\  \kappa_0\in\left(0,\frac{\g}{5}\right],\ \kappa_{\infty}\in(0,\kappa_0), \ p>5d,\ \alpha\in\left(\frac{5}{4},\frac{2p}{p+2d}\right).$$

The main result on Green's function estimates  is

\begin{thm}\label{thm1}
	Let $\mu\in\mathscr{H}(\lambda)$ (i.e.,  $\mathcal{D}_{\lambda}(\mu)>0$), $E_0\in\R$, $L_0\in \N$ and $L_{s+1}=[L_s^{\alpha}]$ ($s\geq 0$). Then for $0<\beta<\mathcal{D}_{\lambda}(\mu)$, there exists
\begin{align*}
	\underline{L}_0=\underline{L}_0(\lambda,\mu, \beta,d,p,\g,\rho,\rho',\alpha,\kappa_0,\kappa_{\infty})>0
\end{align*}
such that the following holds true.   For $L_0\ge\underline{L}_0$, there are $\ep_0=\ep_0(\lambda,\mu,\beta,d,\g,\rho, L_0)>0$ and $\et=\et(\lambda,\mu,\beta,d,L_0)>0$ so that  if $0<\ep<\ep_0$ and $s\ge0$,  then we have for all $\|\bm x-\bm y\|>2L_s$, 
\begin{align*}
	\Po\{\exists E\in[E_0-\et,E_0+\et]\text{ s.t., both $B_{L_s}(\bm x)$ and $B_{L_s}(\bm y)$ are $(\kappa_{\infty},E)$-{\rm bad}}\}\le L_s^{-2p}. 
\end{align*}
\end{thm}

As an  application of the above Green’s function estimates,  we have
\begin{thm}\label{thm2}
	Let $H_{\omega}$ be defined by \eqref{y1} with the common distribution $\mu\in\mathscr{H}(\lambda)$. Fix $\beta\in (0,\mathcal{D}_{\lambda}(\mu))$.
	Then there exist $\varepsilon_0=\varepsilon_0(\lambda,\mu,\beta, d,p,\g,\rho,\rho',\alpha,\kappa_0,\kappa_{\infty})>0$ and $\et=\et(\lambda,\mu,\beta, d,p,\g,\rho,\rho',\alpha,\kappa_0,\kappa_{\infty})>0$ such that for $0<\varepsilon<\varepsilon_0$, $H_\omega$ has pure point spectrum in $[E_0-\et,E_0+\et]$ for $\forall E_0\in\R$ and $\mathbb{P}$ almost every $\omega\in\Omega$. Moreover, for $\mathbb{P}$ almost every $\omega\in\Omega$, 
	there exists a complete system of eigenfunctions $\psi_\omega=\{\psi_{\omega}(\bm x)\}_{\bm x\in{\Z}^d}$ satisfying  
	\begin{align}\label{ww1}
		|\psi_\omega(\bm x)|\le e^{-\frac{\kappa_{\infty}}{2\alpha^{\rho}}\log^{\rho}(1+\|\bm x\|)}\ {\rm for}\  \|\bm x\|\gg 1. 
	\end{align}
	
\end{thm} 
\begin{rem}
 This theorem  requires  the smallness  condition of $\ep_0$,  which depends sensitively on $\lambda,\beta$ (cf. \eqref{et} and \eqref{ep}  in the following for details).  In the proof, we take $\et=4^{-1}(\beta^{-1}\underline{L}_0^p(2\underline{L}_0+1)^d)^{-\frac{1}{\lambda}}$,  which also depends  on $\lambda,\beta$ (thus on $\mu$), where $\underline{L}_0$ is defined in Theorem \ref{thm1}. Since both $\ep_0$ and $\eta$ do  not depend on $E_0$, we can prove indeed that $H_\omega$ has pure point spectrum on $\R$ for $\mathbb{P}$ almost every $\omega\in\Omega$ by  covering  $\R$ with  intervals of length $\et$. 
\end{rem}

\section{The initial step}\label{tis}
In this section, we will prove that the conclusion of Theorem \ref{thm1} holds true for $s=0$ since $0<\ep\ll1$ and $\mu\in\mathscr{H}(\lambda)$.

The following lemma is useful for dealing with matrices with slowly decaying off-diagonal elements.
\begin{lem}\label{qua}
	For $x_i\ge0$, $1\le i\le n$, we have 
	\begin{align}\label{quaeq}
		\log^{\rho}(1+\sum\limits_{i=1}^n x_i)\le \sum\limits_{i=1}^n \log^{\rho}(1+x_i)+C(\rho)\log^{\rho}n,
	\end{align}
	where $C(\rho)>0$   is some constant depending  only on $\rho>0$.
\end{lem}
\begin{proof}
	For a detailed proof, we refer to the Appendix \ref{APPqua}.
\end{proof}
\begin{rem}
	We have the {\it quasi-metric} property: for any $\bm x_i\in\Z^d$ ($1\leq i\leq n$),   
	\begin{align*}
		\log^{\rho}(\|\sum_{i=1}^n\bm x_i\|+1)\leq \sum_{i=1}^n\log^{\rho}(\|\bm x_i\|+1)+C(\rho) \log^{\rho} n.
	\end{align*}
\end{rem}
We have 
\begin{thm}\label{P1}
	Let $\mu\in\mathscr{H}(\lambda)$. Fix $0<\beta<\mathcal{D}_{\lambda}(\mu)$ and  $E_0\in\R$. Then there exists
	\begin{align*}
		\underline{L}_0=\underline{L}_0(\lambda,\mu, \beta,d,p,\g,\rho,\rho',\kappa_0)>0
	\end{align*}
	such that the following holds true.  If $L_0\ge\underline{L}_0$, then there are $\ep_0=\ep_0(\lambda,\mu, \beta,d,p,\g,\rho, L_0)>0$ and $\et=\et(\lambda,\mu,\beta, d,p,L_0)>0$ so that if $0<\ep<\ep_0$,  then we have for all $\|\bm x-\bm y\|>2L_0,$
	\begin{align*}
		\Po\{\exists E\in[E_0-\et,E_0+\et]\text{ s.t.,  both $B_{L_0}(\bm x)$ and $B_{L_0}(\bm y)$ are $(\kappa_0,E)$-{\rm bad}}\}\le L_0^{-2p}.
	\end{align*}
\end{thm}
\begin{proof}
	The proof is based an  application  of the Neumann series expansion argument. 
	Define 
	\begin{align}\label{et}
		\zeta= 2^{-1}\left(\beta^{-1}L_0^{p}(2L_0+1)^d\right)^{-\frac{1}{\lambda}},\ \et=\frac{\zeta}{2}.
	\end{align}
Take $L_0\ge\underline{L}_0\gg1$ so that  $\zeta<\frac{1}{2}$ and $2\zeta\le \beta_0$, where $\beta_0=\beta_0(\mu,\beta)$ is defined in \eqref{vv23}. Define the event
	\begin{align*}
		\textbf{R}_{\bm x}(\zeta):\ \exists\bm z\in B_{L_0}(\bm x)\ {\rm s.t.}, \ |V_{\omega}(\bm z)-E_0|\le \zeta.
	\end{align*}
	From \eqref{vv23}, \eqref{et} and $2\zeta\le \beta_0$,  it follows that 
	\begin{align}\label{L0-p}
		\nonumber\Po(\textbf{R}_{\bm x}(\zeta))&\le (2L_0+1)^d\mu([E_0-\zeta,E_0+\zeta])\\
		\nonumber&\le 2^{\lambda}(2L_0+1)^d\beta^{-1}\zeta^{\lambda}\\
		&\le L_0^{-p}.
	\end{align}
	
	Next, suppose that  $\omega\notin\textbf{R}_{\bm x}(\zeta)$. Then   for $\forall\bm m\in B_{L_0}(\bm x)$ and $\forall |E-E_0|\le \et$, we get
	\begin{align*}
		|V_\omega(\bm m)-E|\ge |V_\omega(\bm m)-E_0|-|E-E_0|\ge \frac{\zeta}{2}.
	\end{align*}
	This implies  that 
	$\|\mathcal{Q}^{-1}\|_2\le \frac{2}{\zeta}$,
	where
	$$\mathcal{Q}=R_{B_{L_0}(\bm x)}(V_{\omega}(\bm m)\delta_{\bm m,\bm n}-E)R_{B_{L_0}(\bm x)}.$$ 
	Note that 
	\begin{align*}
		H_{B_{L_0}(\bm x)}-E=\ep R_{B_{L_0}(\bm x)}\Gamma_{\phi}R_{B_{L_0}(\bm x)}+\mathcal{Q}.
	\end{align*}
	Let 
	\begin{align}\label{ep}
		\ep_0=\ep_0(\lambda,\mu, \beta,d,p,\g,\rho, L_0)=\min\left(\frac{\zeta}{4(\|\Gamma_{\phi}\|_{2}+1)},\frac{\zeta^2}{2(2L_0+1)^d}\right),
	\end{align}
	and assume $0<\ep< \ep_0$.  Then 
	$$ \|\varepsilon \mathcal{Q}^{-1}R_{B_{L_0}(\bm x)}\Gamma_{\phi}R_{B_{L_0}(\bm x)}\|_2\leq \frac12. $$
From  $L_0\ge\underline{L}_0\gg1$ and  the Neumann series expansion argument, 
		it follows that 
	\begin{align}\label{a6}
		\|\mathcal{G}^{E}_{B_{L_0}(\bm x)}\|_2\le 2\|\mathcal{Q}^{-1}\|_2\le 8\left(\frac{L_0^{p}(2L_0+1)^d}{\beta}\right)^{\frac{1}{\lambda}}\le e^{\log^{\rho'}L_0}. 
	\end{align}
	Moreover,   we have  for $\|\bm x'-\bm x''\|\ge L_0^{\frac{4}{5}}$, 
	\begin{align*}
		|\mathcal{G}^{E}_{B_{L_0}(\bm x)}(\bm x',\bm x'')|&\le \sum_{n=1}^{\infty}\left|\ep^n\left(\mathcal{Q}^{-1}\left(\Gamma_{\phi}\mathcal{Q}^{-1}\right)^n\right)(\bm x',\bm x'')\right|\\
		&\le \ep|(\mathcal{Q}^{-1}\Gamma_{\phi}\mathcal{Q}^{-1})(\bm x',\bm x'')|\\
		&\ \ \ +\sum_{n=2}^{\infty}\left(2\zeta^{-1}(2\ep\zeta^{-1})^n\sum_{\bm k_1,\cdots,\bm k_{n-1}\in B_{L_0}(\bm x)}|\Gamma_{\phi}(\bm x',\bm k_1)|\cdots|\Gamma_{\phi}(\bm k_{n-1},\bm x'')|\right)\\
		&\le  4\ep\zeta^{-2}e^{-\gamma \log^{\rho}(1+\|\bm x'-\bm x''\|)}\\
		&\ \ \ +\sum_{n=2}^{\infty}\left(2\zeta^{-1}(2\ep\zeta^{-1}(2L_0+1)^d)^ne^{-\g\log^{\rho}(1+\|\bm x'-\bm x''\|)+\g C(\rho)\log^{\rho} n}\right)\\
		&\le e^{-\kappa_0\log^{\rho}(1+\|\bm x'-\bm x''\|)},
	\end{align*}
	where in the third inequality, we have used  \eqref{y2-d} and  \eqref{qua}. 
	Hence we have shown  that  $B_{L_0}(\bm x)$ is $(\kappa_0,E)$-good for $\omega\notin\textbf{R}_{\bm x}(\zeta)$, namely,  
	\begin{align}\label{807}
		\{\exists E\in[E_0-\et,E_0+\et]\ \text{s.t.,}\ B_{L_0}(\bm x)\ \text{is $(\kappa_0,E)$-bad}\}\subset \textbf{R}_{\bm x}(\zeta).
	\end{align}
	
	Finally, from \eqref{L0-p}, \eqref{807}, $\|\bm x-\bm y\|>2L_0$, $0<\ep<\ep_0$ and the $i.i.d$ assumption on the potentials, we have
	\begin{align*}
	&\ \ \mathbb{P}\{\exists E\in[E_0-\et,E_0+\et]\text{ s.t., both $B_{L_0}(\bm x)$ and $B_{L_0}(\bm y)$ are $(\kappa_0,E)$-bad}\}\\
	\le&\ \ \Po(\textbf{R}_{\bm x}(\zeta))\Po(\textbf{R}_{\bm y}(\zeta))\le  L_0^{-2p}.
	\end{align*}

This completes the proof. 
\end{proof}

\section{The Wegner Estimate}\label{wg}
In order to complete the proof of the Theorem \ref{thm1} via  the MSA induction, we also need the following  important Wegner  estimate. It  has  essentially been proven by   Carmona-Klein-Martinelli \cite{CKM87}, and the regularity property  of  the distribution $\mu$ plays a crucial role there.
\begin{lem}[Wegner estimate, cf. Theorem 6.2 of \cite{CKM87}]\label{x7-0}
	Let $\mu\in\mathscr{H}(\lambda)$ (i.e.,  $\mathcal{D}_{\lambda}(\mu)>0$). Then for any $0<\beta<\mathcal{D}_{\lambda}(\mu)$, we can find $\beta_0=\beta_0(\mu,\beta)>0$ such that
	\begin{align}\label{x8}
		\mathbb{P}\left\{\dist(E,\sigma(H_{B_L(\bm x)}))\le\epsilon\right\}
		\le \beta^{-1}2^{\lambda}(2L+1)^{d(1+\lambda)}\epsilon^{\lambda}
	\end{align} 
	for all $E\in\R$, $\bm x\in\Z^d$, $\epsilon>0$ and $L>0$ with $\epsilon(2L+1)^d\le\beta_0$.
\end{lem}
\begin{proof}
	Note  that the long-range hopping term $\varepsilon \Gamma_{\phi}$ in  our model  is non-random. Then the proof is   similar   to that in the Schr\"odinger case  investigated by Carmona-Klein-Martinelli \cite{CKM87}. We omit the details here.
\end{proof}
\begin{rem}
	This lemma is typically used to provide  desired upper bounds on  the operator norm of  Green's functions in the random operators case. We also want to remark  that  the estimate \eqref{x8} does not depend on   $\varepsilon \Gamma_\phi$. 
\end{rem}

\section{Iteration Theorem}\label{ITthm}

In this section, we introduce  an iteration theorem, which mainly deals  with Green's function estimates in the induction steps.  We   first define the induction parameters  ($s\geq 0$): 
\begin{align}
\nonumber L_0&\in\N,\ L_{s+1}=[L_s^{\alpha}],\ \kappa_0\in\left(0,\frac{\g}{5}\right],\ \kappa_{\infty}\in(0,\kappa_0),\\
	\label{kappa}	\kappa_{s+1}&=\kappa_{s}-\left(\frac{10\g}{L_s^{\frac{4}{5}\alpha-1}}+\frac{10\g}{\log^{\rho-1}L_s}+\frac{20+\alpha^{\rho'}}{\log^{\rho-\rho'}L_s}\right) (1<\rho'<\rho<\rho'+1).
\end{align}
 
\begin{defn}\label{stm}
	Fix $E_0\in\R$.   Let $I=[E_0-\et,E_0+\et]$ with some $\et>0$. For each $s\geq 0,$  we define 
	\begin{itemize}
	\item [${\bf(P1)}_s$:] for all $\|\bm x-\bm y\|>2L_{s}$, 
	\begin{align}\label{P1e}
	\Po\{\exists E\in I\text{ s.t.,  both $B_{L_{s}}(\bm x)$ and $B_{L_{s}}(\bm y)$ are $(\kappa_{s},E)$-bad}\}\le L_{s}^{-2p};
\end{align}
	\item [${\bf(P2)}_s$:]  for all $\|\bm x-\bm y\|> 2L_{s}$,
	\begin{align}\label{P2e}
	\Po\left\{\bigcup_{E\in I}(\textbf{D}_{s,\bm x}(E)\cap \textbf{D}_{s,\bm y}(E))\right\}\le \frac{1}{4}L_{s}^{-2p},
\end{align}
\end{itemize} 
 where for $\bm w=\bm x, \bm y,$
\begin{align}\label{Dsx}
\textbf{D}_{s,\bm w}(E):=\{\text{$\exists U\in \mathscr{T}_{s,\bm w}$ s.t.,  $U$ is $E$-R}\}
\end{align}
with 
{\begin{align}\label{Tsx}
		\mathscr{T}_{s,\bm w}:=\{B_{L}(\bm z):\ B_{L}(\bm z)\subset B_{L_s}(\bm w)\text{ with $L=10L_{s-1},40L_{s-1},130L_{s-1},L_{s}$}\}.
\end{align}}
\end{defn}

\begin{rem}
	The initial estimate ${\bf(P1)}_0$ holds true as shown in Theorem \ref{P1}, which requires the smallness of $\varepsilon$ (i.e., the large disorder condition) in the current setting. This base case serves as the starting point for the MSA. Importantly, as will be demonstrated later, the verification of both ${\bf(P1)}_s$ and  ${\bf(P2)}_s$ for all scales $s\geq 1$ does not require the large disorder condition. 
		Indeed, ${\bf(P2)}_s$ can be  established  via the  { priori}  Wegner estimate (i.e., there is no need to perfrom the multi-scale induction), which is also  independent of the large disorder condition. 
		This key observation suggests that, once the initial scale is controlled under large disorder, the induction mechanism can propagate to all scales regardless of the disorder strength. Consequently, we may expect localization to hold at extreme energies for any non-zero disorder, mirroring the well-known behavior of  Schr\"odinger operators (cf. \cite{CKM87}).   In $d=1$, it is known that the localization holds true on the whole energy interval for all non-zero disorder. This was proved using transfer matrix methods, which are not available for long-range hopping operators. The presence of a long-range hopping makes this  all coupling localization result difficult  to prove   (cf. \cite{JM99} for the   proof of spectral localization  for   one dimensional  random operators with some  power-law long-range  hopping). 
	
\end{rem}

The  iteration theorem  reads as 
\begin{thm}\label{ite}
Let  $L_0\geq \underline{L}_0\gg1$ and $s\geq1$.  
 Assume that both  ${\bf(P1)}_{s-1}$ and ${\bf(P2)}_s$  hold true. Then  ${\bf(P1)}_s$ holds true as well. 
\end{thm}
The proof of this  iteration theorem  will be finished in \S \ref{pot1}. 


\section{The validity of $ ({\bf P2})_s$}\label{sVP2}

In this section,  we  aim to prove the validity of $({\bf P2})_s$ for $s\geq 1$. The proof  will follow directly from the Wegner estimate. 
\begin{thm}[Verification of $({\bf P2})_s$]\label{VP2}
	Let $\mu\in\mathscr{H}(\lambda)$ (i.e.,  $\mathcal{D}_{\lambda}(\mu)>0$). Fix $0<\beta<\mathcal{D}_{\lambda}(\mu)$. Then for $L_s\ge L_{s-1}\ge L_0\ge \underline{L}_0(\lambda,\mu,\beta,d,p,\rho',\alpha)>0$, we have
	$$\Po\left\{\bigcup_{E\in I}(\textbf{D}_{s,\bm x}(E)\cap \textbf{D}_{s,\bm y}(E))\right\}\le \frac{1}{4}L_{s}^{-2p}$$
	for all $\|\bm x-\bm y\|> 2L_{s}$, where $\textbf{D}_{s,\bm x}(E)$ and $\textbf{D}_{s,\bm y}(E)$ are defined in \eqref{Dsx}.
\end{thm}
\begin{proof}
	Suppose that $\omega\in\bigcup_{E\in I}(\textbf{D}_{s,\bm x}(E)\cap \textbf{D}_{s,\bm y}(E))$.  Then there are some $E
	\in I$, $U_{\bm x}\in \mathscr{T}_{s,\bm x}$ and $U_{\bm y}\in \mathscr{T}_{s,\bm y}$  (cf. \eqref{Tsx}) such that
	\begin{align*}
		\dist(E,\sigma(H_{U_{\bm x}}))<e^{-\log^{\rho'}L_{s-1}},\ \dist(E,\sigma(H_{U_{\bm y}}))<e^{-\log^{\rho'}L_{s-1}},
	\end{align*}
	and then 
	\begin{align*}
		\dist(\sigma(H_{U_{\bm x}}),\sigma(H_{U_{\bm y}}))<2e^{-\log^{\rho'}L_{s-1}}.
	\end{align*}
	This implies 
	$$\bigcup_{E\in I}(\textbf{D}_{s,\bm x}(E)\cap \textbf{D}_{s,\bm y}(E))\subset \bigcup_{U_{\bm x}\in \mathscr{T}_{s,\bm x} \atop U_{\bm y}\in \mathscr{T}_{s,\bm y} }\left\{\dist(\sigma(H_{U_{\bm x}}),\sigma(H_{U_{\bm y}}))<2e^{-\log^{\rho'}L_{s-1}}\right\}.$$

	Next, it  needs  to control $\Po\{\dist(\sigma(H_{U_{\bm x}}),\sigma(H_{U_{\bm y}}))\le 2e^{-\log^{\rho'}L_{s-1}}\}$. 
	For this purpose, we will use Lemma \ref{x7-0}.  Then applying Lemma \ref{x7-0} with $\epsilon=2e^{-\log^{\rho'}L_{s-1}}$  and  $L\in\{L_{s},10L_{s-1},40L_{s-1},130L_{s-1}\}$  yields 
	\begin{align*}
		\Po\{\dist(E,\sigma(H_{U_{\bm y}}))\le 2e^{-\log^{\rho'}L_{s-1}}\}\le \beta^{-1}2^{\lambda}(2L_{s}+1)^{d(1+\lambda)}e^{-\lambda\log^{\rho'}L_{s-1}}.
	\end{align*}
	 {From  the $i.i.d$ assumption of  the potentials, $L_s\ge L_{s-1}\ge L_0\ge\underline{L}_0>0$ and $\# {U_{\bm x}}\le (2L_s+1)^d$,  it follows that  (similar to the proof of Lemma 5.28 in \cite{Kir08})}
	\begin{align*}
		&\ \ \Po\{\dist(\sigma(H_{U_{\bm x}}),\sigma(H_{U_{\bm y}}))\le 2e^{-\log^{\rho'}L_{s-1}}\}\\
		\le&\ \  \sum_{i=1}^{\# U_{\bm x}}\Po\{\dist(E,\sigma(H_{U_{\bm y}}))\le 2e^{-\log^{\rho'}L_{s-1}}\}\\
		\le&\ \  \beta^{-1}2^{\lambda}(2L_s+1)^{d(2+\lambda)}e^{-\lambda\log^{\rho'}L_{s-1}}.
	\end{align*}

	Finally, since $\# \mathscr{T}_{s,\bm x},\ \# \mathscr{T}_{s,\bm y}\le 4(2L_s+1)^d$ and $L_s\ge L_{s-1}\ge L_0\ge \underline{L}_0>0$,  we obtain 
	\begin{align*}
		&\ \ \Po\left\{\bigcup_{E\in I}(\bm D_{\bm x}(E)\cap \bm D_{\bm y}(E))\right\}\\
		\le&\ \ \sum_{U_{\bm x}\in \mathscr{T}_{s,\bm x} \atop U_{\bm y}\in \mathscr{T}_{s,\bm y}}\Po\{\dist(\sigma(H_{U_{\bm x}}),\sigma(H_{U_{\bm y}}))\le 2e^{-\log^{\rho'}L_{s-1}}\}\\
		\le&\ \  16\beta^{-1}2^{\lambda}(2L_s+1)^{d(4+\lambda)}e^{-\lambda\log^{\rho'}L_{s-1}}\le \frac{1}{4}L_s^{-2p}.
	\end{align*}

This completes the proof. 

\end{proof}

\section{Coupling Lemma and the Proof of Theorems \ref{thm1}}\label{pot1}
In this section, we aim to  prove Theorem \ref{ite} and  Theorem \ref{thm1}. One of the key  ingredients  toward proving  Theorem \ref{ite} is  a   \textbf{Coupling Lemma}.   

 In the following three subsections, we will prove the \textbf{Coupling Lemma}, Theorem \ref{ite}, and Theorem \ref{thm1}, respectively.

\subsection{{\bf Coupling Lemma}}
We have 
\begin{lem}[{\bf Coupling Lemma}]\label{cl}
	Let $E\in\R$. Assume that 
	\begin{enumerate}
		\item [(1)] $B_{L_{s}}(\bm x)$ is $E$-{\rm NR} {\rm (cf. \eqref{y5})};
		\item [(2)] Each $B_{jL_{s-1}}(\bm z)\subset B_{L_s}(\bm x)$ with  $ j=10,40,130$  is $E$-{\rm NR}; 
		\item[(3)]  
			There are at most three pairwise disjoint $(\kappa_{s-1},E)$-{\rm bad} $L_{s-1}$-cubes in $B_{L_s}(\bm x)$  {\rm(cf. \eqref{y6})}.
	\end{enumerate}  
	Then for 
	\begin{align*}
		L_s\ge L_{s-1}\ge L_0\ge \underline{L}_0(d,\rho,\rho',\alpha,\kappa_0,\kappa_{\infty})>0,
	\end{align*}
	$B_{L_s}(\bm x)$ is $(\kappa_s,E)$-{\rm good} {\rm (cf. \eqref{y6})}.
\end{lem}
\begin{proof}[Proof of Lemma \ref{cl}] 
In order to complete the proof, it  needs  to prove 
	\begin{align}\label{gs}
		|\mathcal{G}_{B_{L_s}(\bm x)}^E(\bm z,\bm y)|\le e^{-\kappa_{s}\log^{\rho}(\|\bm z-\bm y\|+1)}\ \text{for $\|\bm z-\bm y\|\ge L_s^{\frac{4}{5}}$}.
	\end{align}
	The proof is based on the iteration of resolvent identity, and can be divided into   three steps. 	\begin{itemize}
		\item[\textbf{Step 1}:] \textbf{Estimates on good sites}
	\end{itemize}
	We begin with a  geometric construction. 
	\begin{lem}\label{neighbor}
		For any $\bm z\in B_{L_s}(\bm x)$, there is a point $\hat{\bm z}$ such that
		\begin{align}\label{hx}
			\bm z\in B_{L_{s-1}}(\hat{\bm z})\subset B_{L_s}(\bm x)
		\end{align}
		and 
		\begin{align}
			\label{dhx1}\|\bm y-\bm z\|&\ge\|\bm y-\hat{\bm z}\|\ge L_{s-1}+1\ \text{\rm for}\ \bm y\in B_{L_s}(\bm x)\setminus B_{L_{s-1}}(\hat{\bm z}),\\
			\label{dhx2}\|\bm y-\hat{\bm z}\|&\le L_{s-1}\ \text{\rm for}\ \bm y\in B_{L_s}(\bm x)\cap B_{L_{s-1}}(\bm z).
		\end{align}
	\end{lem}
	\begin{proof}
		For a detailed proof, we refer to the Appendix \ref{APPdan}.
	\end{proof}
	 The main result  in  this step is 
\begin{lem}\label{gfeg}
Let $L_s\ge L_{s-1}\ge L_0\ge \underline{L}_0(d,\rho,\rho',\alpha)>0$ and $\bm z\in B_{L_s}(\bm x)$. Assume that $B_{L_{s-1}}(\hat{\bm z})\subset B_{L_s}(\bm x)$ is $(\kappa_{s-1},E)$-$\good$ and $\bm y\in B_{L_s}(\bm x)\setminus B_{L_{s-1}}(\hat{\bm z})$, where $\hat{\bm z}$ is defined in $\eqref{hx}$. Then there is some $\bm z'\in B_{L_s}(\bm x)\setminus B_{L_{s-1}}(\hat{\bm z})$,
	\begin{align}\label{y8}
		|\mathcal{G}_{B_{L_s}(\bm x)}^E(\bm z,\bm y)|\leq e^{-K_{s-1}\cdot\log^{\rho}(1+ \|\bm z-\bm z'\|)}|\mathcal{G}_{B_{L_s}(\bm x)}^E(\bm z',\bm y)|,
	\end{align}
	where $K_{s-1}=\kappa_{s-1}-\frac{2}{\log^{\rho-\rho '}L_{s-1}}$.
\end{lem}
\begin{proof}
	By using the resolvent identity (cf. \cite{Kir08}) and \eqref{hx}, we obtain
	\begin{align}\label{y9}
		\mathcal{G}_{B_{L_{s}}(\bm x)}^E(\bm z,\bm y)=-\ep\sum\limits_{\substack{\bm m\in B_{L_{s-1}}(\hat{\bm z}) \\
				\bm n\in B_{L_{s}}(\bm x)\setminus B_{L_{s-1}}(\hat{\bm z})}}
		\mathcal{G}_{B_{L_{s-1}}(\hat{\bm z})}^E(\bm z,\bm m)\cdot \Gamma_{\phi}(\bm m,\bm n)\cdot \mathcal{G}_{B_{L_{s}}(\bm x)}^E(\bm n,\bm y).
	\end{align}
	Since $B_{L_{s-1}}(\hat{\bm z})$ and $B_{L_s}(\bm x)$ are all  finite sets, there exist $ \bm m'\in B_{L_{s-1}}(\hat{\bm z})$ and $\bm z'\in B_{L_s}(\bm x)\setminus B_{L_{s-1}}(\hat{\bm z})$ satisfying
	\begin{align*}
		&\ \ |\mathcal{G}_{B_{L_{s-1}}(\hat{\bm z})}^E(\bm z,\bm m')\cdot \Gamma_{\phi}(\bm m',\bm z')\cdot \mathcal{G}_{B_{L_s}(\bm x)}^E(\bm z',\bm y)|\\
		=&\ \ \max\limits_{\substack{\bm m\in B_{L_{s-1}}(\hat{\bm z}) \\
				\bm n\in B_{L_s}(\bm x)\setminus B_{L_{s-1}}(\hat{\bm z})}}|\mathcal{G}_{B_{L_{s-1}}(\hat{\bm z})}^E(\bm z,\bm m)\cdot \Gamma_{\phi}(\bm m,\bm n)\cdot \mathcal{G}_{B_{L_s}(\bm x)}^E(\bm n,\bm y)|.
	\end{align*}
	Therefore, from \eqref{y2-d} and \eqref{y9}, we get
	\begin{align}
		\nonumber|\mathcal{G}_{B_{L_s}(\bm x)}^E(\bm z,\bm y)|\le&\  (2L_{s-1}+1)^{d}(2L_s+1)^{d}|\mathcal{G}_{B_{L_{s-1}}(\hat{\bm z})}^E(
			\bm z,\bm m')|e^{-\gamma\log^{\rho}(1+\|\bm m'-\bm z'\|)}\\
			\label{y10}&\ \ \cdot|\mathcal{G}_{B_{L_s}(\bm x)}^E(\bm z',\bm y)|.
	\end{align}
	Then  we   break the discussion  into two cases.  
	\begin{enumerate}
		\item[\textbf{Case 1}:] $\|\bm m'-\bm z\|>L_{s-1}^{\frac{4}{5}}$. From $B_{L_{s-1}}(\hat{\bm z})$ is $(\kappa_{s-1},E)$-good (cf.  \eqref{y6}) and \eqref{qua}, we have
	\end{enumerate}
	\begin{align}
			\nonumber&\ \ |\mathcal{G}_{B_{L_{s-1}}(\hat{\bm z})}^E(\bm z,\bm m')|\cdot e^{-\gamma\log^{\rho}(\|\bm m'-\bm z'\|+1)}\\
			\nonumber\le&\ \  e^{-\gamma\log^{\rho}(\|\bm z-\bm m'\|+1)}\cdot e^{-\gamma\log^{\rho}(\|\bm m'-\bm z'\|+1)}\\
			\label{y11}\le&\ \  e^{-\gamma\log^{\rho}(\|\bm z-\bm z'\|+1)}\cdot e^{\gamma C(\rho)\log^{\rho}2}.
	\end{align}
	Combining \eqref{y10}, \eqref{y11} and $L_s\ge L_{s-1}\ge L_0\ge\underline{L}_0$ implies (since $\|\bm z'-\bm z\|\geq L_{s-1}+1$ by \eqref{dhx1})
	\begin{align}\label{g1}
	|\mathcal{G}_{B_{L_s}(\bm x)}^E(\bm z,\bm y)|\le e^{-\left(\gamma-\frac{2}{\log^{\rho-\rho'}L_{s-1}}\right)\cdot\log^{\rho}(\|\bm z-\bm z'\|+1)}\cdot|\mathcal{G}_{B_{L_s}(\bm x)}^E(\bm z',\bm y)|.
	\end{align}
	\begin{enumerate}
		\item[\textbf{Case 2}:] $\|\bm m'-\bm z\|\le L_{s-1}^{\frac{4}{5}}$. From \eqref{qua}, we have
		\begin{align}\label{g21}
			e^{-\gamma\log^{\rho}(\|\bm m'-\bm z'\|+1)}\le e^{-\gamma\log^{\rho}(\|\bm z'-\bm z\|+1)}\cdot e^{\gamma\log^{\rho}(\|\bm z-\bm m'\|+1)}\cdot e^{\gamma C(\rho)\log^{\rho}2}.
		\end{align}
		Next, from   $\|\bm m'-\bm z\|\le L_{s-1}^{\frac{4}{5}} \leq \|\bm z'-\bm z\|^{\frac 45}$, it  follows  that
		\begin{align}\label{g22}
			e^{\gamma\log^{\rho}(\|\bm z-\bm m'\|+1)}\le e^{\gamma (\frac{4}{5})^{\rho}\left(\log^{\rho}(\|\bm z'-\bm z\|+1)+2\rho\log 2\log^{\rho-1}(\|\bm z'-\bm z\|+1)\right)}.
		\end{align}
		Combining  \eqref{y5}, \eqref{g21}, \eqref{g22}, $1<\rho'<\rho<\rho'+1$ and $L_s\ge L_{s-1}\ge L_0\ge\underline{L}_0$ implies (since $\|\bm z'-\bm z\|\geq L_{s-1}+1$ by \eqref{dhx1})
			\begin{align}
			\nonumber&\ |\mathcal{G}_{B_{L_{s-1}}(\hat{\bm z})}^E(\bm z,\bm m')|\cdot e^{-\gamma\log^{\rho}(\|\bm m'-\bm z'\|+1)}\\
			\nonumber\le&\ \  e^{\log^{\rho '}L_{s-1}}\cdot e^{-\gamma\log^{\rho}(\|\bm z'-\bm z\|+1)}\cdot e^{\gamma C(\rho)\log^{\rho}2}\\
			\nonumber&\ \ \cdot e^{\gamma (\frac{4}{5})^{\rho}\left(\log^{\rho}(\|\bm z'-\bm z\|+1)+2\rho\log 2\log^{\rho-1}(\|\bm z'-\bm z\|+1)\right)}\\
			\label{806}\le&\ \  e^{-\left(\gamma-\gamma(\frac{4}{5})^{\rho}-\frac{1}{\log^{\rho-\rho'}L_{s-1}}\right)\log^{\rho}(\|\bm z'-\bm z\|+1)}.
		\end{align}
		From \eqref{y10}, \eqref{806} and and $L_s\ge L_{s-1}\ge L_0\ge\underline{L}_0$, we get
		\begin{align}
			\nonumber&\ |\mathcal{G}_{B_{L_s}(\bm x)}^E(\bm z,\bm y)|\\
			\nonumber\le&\  e^{-\left(\gamma-\gamma(\frac{4}{5})^{\rho}-\frac{2}{\log^{\rho-\rho'}L_{s-1}}\right)\log^{\rho}(\|\bm z'-\bm z\|+1)}|\mathcal{G}_{B_{L_s}(\bm x)}^E(\bm z',\bm y)|\\
			\label{g2}\le&\ e^{-\left(\frac{\gamma}{5}-\frac{2}{\log^{\rho-\rho'}L_{s-1}}\right)\log^{\rho}(\|\bm z'-\bm z\|+1)}|\mathcal{G}_{B_{L_s}(\bm x)}^E(\bm z',\bm y)|.
		\end{align}
	\end{enumerate}
	From  $0<\kappa_{s-1}\le \kappa_0\le\frac{\gamma}{5}$ and $L_s\ge L_{s-1}\ge L_0\ge \underline{L}_0>0$,   \eqref{g1} and \eqref{g2}, it follows  that
	\begin{align}\label{y12-1}
		|\mathcal{G}_{B_{L_s}(\bm x)}^E(\bm z,\bm y)|\leq e^{-(\kappa_{s-1}-\frac{2}{\log^{\rho-\rho '}L_{s-1}})\cdot \log^{\rho}(\|\bm z-\bm z'\|+1)}|\mathcal{G}_{B_{L_s}(\bm x)}^E(\bm z',\bm y)|.
	\end{align}
	\end{proof}
		\begin{itemize}
		\item[\textbf{Step 2}]: \textbf{Estimates on bad sites}
	\end{itemize}
 In the following,  we only  deal with  the case that $B_{L_s}(\bm x)$ contains  exactly three pairwise disjoint $(\kappa_{s-1},E)$-bad $L_{s-1}$-cubes $B_{L_{s-1}}(\bm w_i)$, $1\le i\le 3$, since the other cases are easier to handle. 
\begin{lem}\label{3kuai}
	For $1\le i\le 3$, there is some $\bm w_i^*\in\Z^d$ such that
	\begin{align*}
		B_{L_{s-1}}(\bm w_i)\subset B_{10L_{s-1}}(\bm w_i^{*})\subset B_{L_s}(\bm x)
	\end{align*}
	and
	\begin{align}
			\label{dhx3}\|\bm y-\bm w_i\|&\ge\|\bm y-\bm w_{i}^{*}\|\ge 10L_{s-1}+1\ \text{\rm for}\ \bm y\in B_{L_s}(\bm x)\setminus B_{10L_{s-1}}(\bm w_i^{*}).
\end{align}
	Moreover, if $\bm z\in B_{L_s}(\bm x)\setminus\cup_{i=1}^3 B_{10L_{s-1}}(\bm w_i^{*})$, then $B_{L_{s-1}}(\hat{\bm z})$ is $(\kappa_{s-1},E)$-{\rm good},  where $\hat{\bm z}$ is defined in \eqref{hx}.
\end{lem}
\begin{proof}
	For a detailed proof, we refer to the Appendix \ref{APPdan}.
\end{proof}

\begin{lem}\label{dancub}
There are {cubes} $B_{1},B_2,B_3\subset B_{L_{s}}(\bm x)$  satisfying
	\begin{enumerate}
		\item [(1).] ${\rm dist}(B_{i},B_{j})\ge 10 L_{s-1},\ 1\le i\neq j\le 3$; 
		\item [(2).] $\Pi:= \bigcup\limits_{i=1}^{3}B_{i}\supset\bigcup\limits_{i=1}^{3}B_{10L_{s-1}}(\bm w_{i}^{*})$; 
		\item [(3).] $\sum\limits_{i=1}^{3}l_{i}\leq 260L_{s-1}$, where $l_i=\diam(B_i)$; 
		\item [(4).] $B_i$ is $E$-{\rm NR}  for $1\le i\le 3$, {\rm (cf. \eqref{y5})};
		\item [(5).] $B_{L_{s-1}}(\hat{\bm z})$ is $(\kappa_{s-1},E)$-{\rm good} for $\bm z\in B_{L_s}(\bm x)\setminus\Pi$,  where $\hat{\bm z}$ is defined in \eqref{hx}.
	\end{enumerate}
\end{lem}

\begin{proof}
	For a detailed proof, we refer to the Appendix \ref{APPdan}.
\end{proof}
The main result  in this step is

\begin{lem}\label{z101}
	Let $L_s\ge L_{s-1}\ge L_0\ge \underline{L}_0(d,\g,\rho,\rho',\alpha,\kappa_{0},\kappa_{\infty})>0$ and fix $1\le i\le 3$. Assuming  $\bm z\in B_{i}$ and $\bm y\in B_{L_s}(\bm x)\setminus B_{i}$,  then  there exist some $\tilde{\bm z}\in B_{i}$ and $ {\bm z'}\in B_{L_s}(\bm x)\setminus B_{i}$ such that
	\begin{align}\label{z102}
		|\mathcal{G}_{B_{L_s}(\bm x)}^E(\bm z,\bm y)|\leq e^{-(\kappa_{s-1}-\frac{20}{\log^{\rho-\rho '}L_{s-1}})\log^{\rho}(\|\tilde{\bm z}-\bm z'\|+1)}|\mathcal{G}_{B_{L_s}(\bm x)}^E(\bm z',\bm y)|.
	\end{align}
\end{lem}
\begin{proof}
	Using again the resolvent identity yields 
	\begin{align}
		\mathcal{G}_{B_{L_s}(\bm x)}^E(\bm z,\bm y)=-\ep \sum\limits_{\substack{
				\bm m\in B_{i}\\
				\bm n\in B_{L_s}(\bm x)\setminus B_{i}}}    \mathcal{G}_{B_{i}}^E(\bm z,\bm m)\cdot \Gamma_{\phi}(\bm m,\bm n)\cdot \mathcal{G}_{B_{L_s}(\bm x)}^E(\bm n,\bm y).
	\end{align}
	Similar to the proof of  \eqref{y10}, we can find $\tilde{\bm z}\in B_{i}$ and $\bm z'\in B_{L_s}(\bm x)\setminus B_{i} $ such that
	\begin{align*}
			|\mathcal{G}_{B_{L_s}(\bm x)}^E(\bm z,\bm y)|
			\leq (l_{i}+1)^{d}(2L_s+1)^{d}\cdot|\mathcal{G}_{B_{i}}^E(\bm z,\bm \tilde{\bm z})|\cdot e^{-\gamma\log^{\rho}(\| \tilde{\bm z}-\bm z'\|+1)}|\mathcal{G}_{B_{L_s}(\bm x)}^E(\bm z',\bm y)|.
	\end{align*}
	Since $\sum_{i=1}^{3}l_i\le 260L_{s-1}$ and $B_{i}$ is $E$-NR (cf. \eqref{y5}), we have 
	\begin{align}
		\nonumber	|\mathcal{G}_{B_{L_s}(\bm x)}^E(\bm z,\bm y)|\le&\  (260L_{s-1}+1)^{d}(2L_s+1)^{d}\cdot e^{\log^{\rho'}(260L_{s-1})}\\
		\label{b11}	&\ \ \cdot e^{-\gamma\log^{\rho}(\|\tilde{\bm z}-\bm z'\|+1)}|\mathcal{G}_{B_{L_s}(\bm x)}^E(\bm z',\bm y)|.
	\end{align}
	
	We again  divide   the discussion   into two cases, $\|\tilde{\bm z}-\bm z'\|\ge L_{s-1}$ and $\|\tilde{\bm z}-\bm z'\|< L_{s-1}$.
	\begin{enumerate}
		\item [\textbf{Case 1}:] $\|\tilde{\bm z}-\bm z'\|\ge L_{s-1}$. From $L_s\ge L_{s-1}\ge L_0\ge \underline{L}_0>0$ and $\|\tilde{\bm z}-\bm z'\|\ge L_{s-1}$,  we get 
	\begin{align}
		\nonumber&\ \ (260L_{s-1}+1)^{d}(2L_s+1)^{d}\cdot e^{\log^{\rho'}(260L_{s-1})}\cdot e^{-\gamma\log^{\rho}(\|\tilde{\bm z}-\bm z'\|+1)}\\
	\label{b12}	\le&\ \  e^{-\left(\gamma-\frac{3}{\log^{\rho- \rho '}L_{s-1}}\right)\log^{\rho}(\|\tilde{\bm z}-\bm z'\|+1)}.
	\end{align}
	Thus,  combining \eqref{b11} and \eqref{b12} gives   
		\begin{align*}
				|\mathcal{G}_{B_{L_s}(\bm x)}^E(\bm z,\bm y)|\le e^{-\left(\gamma-\frac{3}{\log^{\rho- \rho '}L_{s-1}}\right)\log^{\rho}(\|\tilde{\bm z}-\bm z'\|+1)}|\mathcal{G}_{B_{L_s}(\bm x)}^E(\bm z',\bm y)|.
		\end{align*}
		\item [\textbf{Case 2}:] $\|\tilde{\bm z}-\bm z'\|< L_{s-1}$. In this case, from \eqref{dhx2} and Lemma \ref{dancub} (5), we have $\tilde{\bm z}\in B_{L_{s-1}}(\hat{\bm z}')$ and $B_{L_{s-1}}(\hat{\bm z}')$ is $(\kappa_{s-1},E)$-good. By Lemma \ref{gfeg}, there is some $\bm z''\in B_{L_s}(\bm x)\setminus B_{L_{s-1}}(\hat{\bm z}')$ such that
		\begin{align}\label{b21}
			|\mathcal{G}_{B_{L_s}(\bm x)}^E(\bm z',\bm y)|\le e^{-K_{s-1}\log^{\rho}(\|\bm z'-\bm z''\|+1)}\cdot |\mathcal{G}_{B_{L_s}(\bm x)}^E(\bm z'',\bm y)|.
		\end{align}
		Since $\bm z''\in B_{L_s}(\bm x)\setminus B_{L_{s-1}}(\hat{\bm z}')$, \eqref{dhx1} and $L_s\ge L_{s-1}\ge  L_0\ge \underline{L}_0>0$, we obtain 
		\begin{align}
			\nonumber&\ \ (260L_{s-1}+1)^{d}(2L_s+1)^{d}\cdot e^{\log^{\rho '}(260L_{s-1})}\cdot e^{-K_{s-1}\log^{\rho}(\|\bm z'-\bm z''\|+1)}\\
			\nonumber\leq&\  \ e^{-2\log^{\rho'}L_{s-1}}\cdot e^{-\left(K_{s-1}-\frac{10}{\log^{\rho-\rho'}L_{s-1}}\right)\log^{\rho}(\|\bm z'-\bm z''\|+1)}\\
			\label{b22}\leq&\ \  e^{-2\log^{\rho'}L_{s-1}}\cdot e^{-\left(\kappa_{s-1}-\frac{20}{\log^{\rho-\rho'}L_{s-1}}\right)\log^{\rho}(\|\bm z'-\bm z''\|+1)}.
		\end{align}
		By \eqref{qua} and $0<\kappa_{s-1}\le \frac{\g}{5}$, we have
		\begin{align}
			\nonumber&\ \ e^{-\gamma\log^{\rho}(\|\tilde{\bm z}-\bm z'\|+1)}\cdot e^{-\left(\kappa_{s-1}-\frac{20}{\log^{\rho-\rho'}L_{s-1}}\right)\log^{\rho}(\|\bm z'-\bm z''\|+1)}\\
		\nonumber	\le&\ \ e^{-\left(\kappa_{s-1}-\frac{20}{\log^{\rho-\rho'}L_{s-1}}\right)\left(\log^{\rho}(\|\tilde{\bm z}-\bm z'\|+1)+\log^{\rho}(\|\bm z'-\bm z''\|+1)\right)}\\
			\label{b23}\le &\ \ e^{K'_{s-1}C(\rho)\log^{\rho}2} e^{-K'_{s-1}\log^{\rho}(\|\tilde{\bm z}-\bm z''\|+1)},
		\end{align}
			where
		\begin{align}\label{K'}
			K'_{s-1}=\kappa_{s-1}-\frac{20}{\log^{\rho-\rho'}L_{s-1}}.
		\end{align}
		According to \eqref{b11}, \eqref{b21}, \eqref{b22} and \eqref{b23}, we obtain 
		\begin{align}
			\nonumber&\ \ |\mathcal{G}_{B_{L_s}(\bm x)}^E(\bm z,\bm y)|\\
			\nonumber\le&\ \  e^{-2\log^{\rho'}L_{s-1}}\cdot e^{K'_{s-1}C(\rho)\log^{\rho}2} e^{-K'_{s-1}\log^{\rho}(\|\tilde{\bm z}-\bm z''\|+1)}\cdot |\mathcal{G}_{B_{L_s}(\bm x)}^E(\bm z'',\bm y)|\\
			\label{z103}\le&\ \  e^{-\log^{\rho'}L_{s-1}} e^{-K'_{s-1}\log^{\rho}(\|\tilde{\bm z}-\bm z''\|+1)}\cdot|\mathcal{G}_{B_{L_s}(\bm x)}^E(\bm z'',\bm y)|.
		\end{align}
		If $\bm z''\notin B_{i}$, Lemma \ref{z101} has been proven. Otherwise,  $\bm z''\in B_{i}$,  and  by a similar argument, there exist some $\tilde{\bm z}''\in B_{i}$ and $ \bm z'''\in B_{L_s}(\bm x)$ such that
		\begin{align*}
			|\mathcal{G}_{B_{L_s}(\bm x)}^E(\bm z,\bm y)|&\leq e^{-2\log^{\rho '}L_{s-1}}\cdot|\mathcal{G}_{B_{L_s}(\bm x)}^E(\bm z''',\bm y)|\\
			&\ \ \cdot e^{-K'_{s-1}\left(\log^{\rho}(\|\tilde{\bm z}-\bm z''\|+1)+\log^{\rho}(\|\tilde{\bm z}''-\bm z'''\|+1)\right)}.
		\end{align*}
		This above iteration procedure   will stop  after  finite many steps until we get (\ref{z102}), and otherwise,
		$|\mathcal{G}_{B_{L_s}(\bm x)}^E(\bm z,\bm y)|$ must vanish. This finishes the proof. 
			\end{enumerate}
\end{proof}

\begin{itemize}
	\item[\textbf{Step 3}]: \textbf{Completion of the proof of Coupling Lemma}
\end{itemize}
	
	In the following, we always assume that $\|\bm z-\bm y\|\ge L_s^{\frac{4}{5}}$. We will prove \eqref{gs} via the iteration using estimates obtained in the above two steps. For $\bm z\in B_{L_s}(\bm x)$, we define 
	\begin{align*}
	O(\bm z)=\left\{\begin{array}{ll}
			B_{i}, & \bm z\in B_i,\\
			B_{L_{s-1}}(\hat{\bm z}),&\bm z\notin\Pi,
		\end{array}\right.
	\end{align*}
where $\hat{\bm z}$ is defined in \eqref{hx}. From Lemma \ref{gfeg} and Lemma \ref{z101}, there are $\tilde{\bm z}\in O(\bm z)$ and $\bm z'\in B_{L_s}(\bm x)\setminus O(\bm z)$ such that
\begin{align}\label{itera}
	|\mathcal{G}_{B_{L_s}(\bm x)}^E(\bm z,\bm y)|\leq e^{-K'_{s-1}\log^{\rho}(\|{\bm z'}- \tilde{\bm z}\|+1)}|\mathcal{G}_{B_{L_s}(\bm x)}^E(\bm z',\bm y)|,
\end{align}
where  $K'_{s-1}=\kappa_{s-1}-\frac{20}{\log^{\rho-\rho'}L_{s-1}}$. 

Next,  iterating  \eqref{itera} for $m\geq 2$  steps  leads to the  following:  
there exist $\bm z_{1}, \bm z_{2}, \cdots, \bm z_{m}\in B_{L_s}(\bm x)$ and $\tilde{\bm z}_0, \tilde{\bm z}_1,\cdots, \tilde{\bm z}_{m-1}\in B_{L_s}(\bm x)$ such that, 
\begin{align}
	\nonumber	&\ \ |\mathcal{G}_{B_{L_s}(\bm x)}^E(\bm z,\bm y)|\\
	\label{vv5}	\le&\ \  e^{-K'_{s-1}\left(\sum_{k=0}^{m-1}\log^{\rho}(\|\bm z_{k+1}-\tilde{\bm z}_{k}\|)+1\right)}|\mathcal{G}_{B_{L_s}(\bm x)}^E(\bm z_{m},\bm y)|,
\end{align}
where $\bm z_0=\bm z$ and $\bm z_{k+1}=\bm z'_k$, $k=0,1,\cdots,m-1$. We define $n\geq 1$ to be the smallest integer so that $\bm z_n\in O(\bm y).$ We then have $\bm z_{i}\notin O(\bm y)$ for $ i=0,1,\cdots, n-1$.  
We divide the discussion into two cases: 
\begin{itemize}
	\item[\textbf{Case 1}:]  $n\le 2\frac{\kappa_0\log^{\rho}(1+\|\bm z-\bm y\|)+\log^{\rho'}L_s}{K'_{s-1}\log^{\rho}(L_{s-1}+1)}+2$. 
Using  Lemma \ref{qua} and  the triangle inequality implies 
 \begin{align}
 		\nonumber&\ \ \sum_{k=0}^{n-1}\log^{\rho}(\|\bm z_{k+1}-\tilde{\bm z}_{k}\|+1)\\
 		\ge&\ \ 
 		\nonumber\log^{\rho}\left(\sum_{k=0}^{n-1}\|\bm z_{k+1}-\tilde{z}_{k}\|+1\right)-C(\rho)\log^{\rho}n\\
 	\nonumber\ge&\ \  \log^{\rho}\left(\left(\sum_{k=0}^{n-1}\|\bm z_{k+1}-\bm z_{k}\|\right)+1-\left(\sum_{k=0}^{n-1}\|\bm z_k-\tilde{\bm z}_k\|\right)\right)-C(\rho)\log^{\rho} n\\
 	\label{801}\ge&\ \ \log^{\rho}\left(\|\bm z_n-\bm z\|+1-\left(\sum_{k=0}^{n-1}\|\bm z_k-\tilde{\bm z}_k\|\right)\right)-C(\rho)\log^{\rho} n.
 \end{align}
	In this  case, since $\bm z_k,\tilde{\bm z}_k\in O(\bm z_k)$ and $\diam (O(\bm z_k))\le 260L_{s-1}$ ($0\le k\le n-1$), we have
\begin{align}\label{z501}
\sum_{k=0}^{n-1}\|\bm z_{k}-\tilde{\bm z}_{k}\|\le 260nL_{s-1}
\end{align}
and
\begin{align}\label{z502}
	\|\bm z_n-\bm z\|\ge \|\bm z-\bm y\|-\|\bm y-\bm z_n\|\ge \|\bm z-\bm y\|-130L_{s-1}. 
\end{align}
 It  follows from \eqref{z501} and \eqref{z502} that 
\begin{align}\label{802}
	\|\bm z_n-\bm z\|-\left(\sum_{k=0}^{n-1}\|\bm z_k-\tilde{\bm z}_k\|\right)\ge\|\bm z-\bm y\|-130L_{s-1}-260nL_{s-1}.
\end{align}
 We then need an elementary  inequality to extract the factor $\log^{\rho}(\|\bm z-\bm y\|+1)$ from \eqref{801}. Indeed,   if $\theta,Q>0$ and $\theta\ll1\ll Q$, then 
\begin{align}
	\nonumber\ \ \left(\log(1-\theta)+Q\right)^{\rho}&=\ \ \left(\frac{\log(1-\theta)}{Q}+1\right)^{\rho}\cdot Q^{\rho}\\
	\nonumber&\ge(1+\rho\frac{\log(1-\theta)}{Q})\cdot Q^{\rho}\\
	\label{ele}&\ge(1-2\rho\frac{\theta}{Q})\cdot Q^{\rho}.
\end{align}
Since  $\|\bm z-\bm y\|\ge L_s^{\frac{4}{5}}\gg1$, applying  \eqref{ele} with $\theta=\frac{130L_{s-1}+260nL_{s-1}}{\|\bm z-\bm y\|+1}<\frac{130L_{s-1}+260nL_{s-1}}{L_s^{\frac{4}{5}}}\ll 1$ and $Q=\log(\|\bm z-\bm y\|+1)$ gives 
\begin{align}
		\nonumber&\ \ \log^{\rho}(\|\bm z-\bm y\|+1-130L_{s-1}-260nL_{s-1})\\
	\nonumber	=&\ \ \left(\log\left(1-\frac{130L_{s-1}+260nL_{s-1}}{\|\bm z-\bm y\|+1}\right)+\log(\|\bm z-\bm y\|+1)\right)^{\rho}\\
	\label{803}	\ge&\ \ \left(1-2\rho\frac{130L_{s-1}+260nL_{s-1}}{\log\left(1+L_s^{\frac{4}{5}}\right)L_s^{\frac{4}{5}}}\right)\log^{\rho}(\|\bm z-\bm y\|+1).
\end{align}
Then combining  \eqref{801}, \eqref{802} and \eqref{803}  shows 
\begin{align}
\nonumber	&\ \ \sum_{k=0}^{n-1}\log^{\rho}(\|\bm z_{k+1}-\tilde{\bm z}_{k}\|+1)\\
	\nonumber\ge&\ \  \log^{\rho}(\|\bm z-\bm y\|+1-130L_{s-1}-260nL_{s-1})-C(\rho)\log^{\rho}n\\
\label{804}	\ge&\ \  \left(1-2\rho\frac{130L_{s-1}+260nL_{s-1}}{\log\left(1+L_s^{\frac{4}{5}}\right)L_s^{\frac{4}{5}}}-\frac{C(\rho)\log^{\rho}n}{\left(\frac{4}{5}\right)^{\rho}\log^{\rho} L_s}\right)\log^{\rho}(\|\bm z-\bm y\|+1). 
\end{align}
From \eqref{vv5}, \eqref{804} and  since $B_{L_s}(\bm x)$ is $E$-NR (cf. \eqref{y5}), we have
\begin{align}
		\nonumber&\ \ |\mathcal{G}_{B_{L_s}(\bm x)}^E(\bm z,\bm y)|\\
		\nonumber\le&\ \  e^{-\left(K'_{s-1}\left(1-2\rho\frac{130L_{s-1}+260nL_{s-1}}{\log(1+L_s^{\frac{4}{5}})L_s^{\frac{4}{5}}}-\frac{C(\rho)\log^{\rho}n}{(\frac{4}{5})^{\rho}\log^{\rho}L_s}\right)\right)\log^{\rho}(\|\bm z-\bm y\|+1)}\cdot e^{\log^{\rho'}L_s}\\
		\le&\ \  \label{805}e^{-\left(K'_{s-1}\left(1-2\rho\frac{130L_{s-1}+260nL_{s-1}}{\log(1+L_s^{\frac{4}{5}})L_s^{\frac{4}{5}}}-\frac{C(\rho)\log^{\rho}n}{(\frac{4}{5})^{\rho}\log^{\rho}L_s}\right)-\frac{1}{(\frac{4}{5})^{\rho}\log^{\rho-\rho'}L_s}\right)\log^{\rho}(\|\bm z-\bm y\|+1)}.
\end{align}
Denote 
\begin{align}\label{K''}
	K''_{s-1}=K'_{s-1}\left(1-2\rho\frac{130L_{s-1}+260nL_{s-1}}{\log(1+L_s^{\frac{4}{5}})L_s^{\frac{4}{5}}}-\frac{C(\rho)\log^{\rho}n}{(\frac{4}{5})^{\rho}\log^{\rho}L_s}\right)-\frac{1}{(\frac{4}{5})^{\rho}\log^{\rho-\rho'}L_s}.
\end{align}
To get \eqref{gs}, it suffices to   prove $K''_{s-1}>\kappa_s.$
Since $\kappa_{s-1}>\kappa_{\infty}$ and $L_{s-1}\ge L_0\ge\underline{L}_0\gg1$, we obtain
\begin{align*}
	K'_{s-1}=\kappa_{s-1}-\frac{20}{\log^{\rho-\rho'}L_{s-1}}>\frac{\kappa_{\infty}}{2}.
\end{align*}
 From $K'_{s-1}>\frac{\kappa_{\infty}}{2}$, $\|\bm z-\bm y\|\le 2L_s$ and $L_s=[L_{s-1}^{\alpha}]$,  it follows that
\begin{align}\label{n}
	\nonumber n &\le 2\frac{\kappa_0\log^{\rho}(1+\|\bm z-\bm y\|)+\log^{\rho'}L_s}{K'_{s-1}\log^{\rho}(L_{s-1}+1)}+2\\
	&\le \frac{6\kappa_0\alpha^{\rho}}{K'_{s-1}}+2\le\frac{12\kappa_0\alpha^{\rho}}{\kappa_{\infty}}+2.
\end{align}
Finally, from $\kappa_{s-1}\le \kappa_0\le \frac{\g}{5}$, $\alpha\in\left(\frac{5}{4},\frac{2p}{p+2d}\right)$, \eqref{kappa}, \eqref{K'} and \eqref{n},  we have for $L_s\ge L_{s-1}\ge L_0\ge \underline{L}_0>0$,
\begin{align}
		\nonumber K''_{s-1}&= \left(\kappa_{s-1}-\frac{20}{\log^{\rho-\rho'}L_{s-1}}\right)\left(1-2\rho\frac{130L_{s-1}+260nL_{s-1}}{\log(1+L_s^{\frac{4}{5}})L_s^{\frac{4}{5}}}-\frac{C(\rho)\log^{\rho}n}{(\frac{4}{5})^{\rho}\log^{\rho}L_s}\right)-\frac{1}{(\frac{4}{5})^{\rho}\log^{\rho-\rho'}L_s}\\
	\nonumber&\ge \kappa_{s-1}-\left(\kappa_{s-1}\left(\frac{260\rho+520n\rho}{\log L_{s}\cdot L_{s-1}^{\frac{4}{5}\alpha-1}}+\frac{C(\rho)\log^{\rho}n}{\log^{\rho}L_{s-1}}\right)+\frac{20+\alpha^{\rho'}}{\log^{\rho-\rho'}L_{s-1}}\right)\\
		\nonumber&\ge \kappa_{s-1}-\left(\frac{50\kappa_{s-1}}{L_{s-1}^{\frac{4}{5}\alpha-1}}+\frac{50\kappa_{s-1}}{\log^{\rho-1}L_{s-1}}+\frac{20+\alpha^{\rho'}}{\log^{\rho-\rho'}L_{s-1}}\right)\\
		\label{vv10}&\ge \kappa_{s-1}-\left(\frac{10\g}{L_{s-1}^{\frac{4}{5}\alpha-1}}+\frac{10\g}{\log^{\rho-1}L_{s-1}}+\frac{20+\alpha^{\rho'}}{\log^{\rho-\rho'}L_{s-1}}\right)=\kappa_s.
\end{align}
We finish the  proof   of \eqref{gs} in this case. 

\item[\textbf{Case 2}:] $n> 2\frac{\kappa_0\log^{\rho}(1+\|\bm z-\bm y\|)+\log^{\rho'}L_s}{K'\log^{\rho}(L_{s-1}+1)}+2$. In order to prove  \eqref{gs} in this case, it suffices to show  
\begin{align}\label{countn}
\# \mathscr{N}\ge\left[\frac{n}{2}\right],
\end{align}
 where
 \begin{align*}
 	\mathscr{N}=\{1\le k\le n-1:\ \|\bm z_{k+1}-\tilde{\bm z}_{k}\|>L_{s-1}\}.
 \end{align*} 
Indeed, from  \eqref{vv5}, \eqref{countn} and $\kappa_0\le \kappa_s$, it follows  that
\begin{align*}
	&\ \ |\mathcal{G}_{B_{L_s}(\bm x)}^E(\bm z,\bm y)|\\
	\le&\ \  e^{-K'_{s-1}\left(\sum_{k=0}^{n-1}\log^{\rho}(\|\bm z_{k+1}-\tilde{\bm z}_{k}\|+1)\right)}|\mathcal{G}_{B_{L_s}(\bm x)}^E(\bm z_{n},\bm y)|\\
		\le&\ \  e^{-\left[\frac{n}{2}\right]K'_{s-1}\log^{\rho}(L_{s-1}+1)}e^{\log^{\rho'}L_s}\\
		\le&\ \  e^{-\kappa_0\log^{\rho}(\|\bm z-\bm y\|+1)}\le e^{-\kappa_s\log^{\rho}(\|\bm z-\bm y\|+1)}.
\end{align*}
Finally, for the  proof of \eqref{countn}, we refer to the Appendix \ref{CN}.

\end{itemize}
This concludes the proof of the {\bf Coupling Lemma}.
\end{proof}

\subsection{Proof  of Theorem \ref{ite}}
In this subsection, we will  prove the iteration theorem (cf.  Theorem \ref{ite}), which follows from combining  the {\bf Coupling Lemma} and the probability estimates. 
\begin{proof}[Proof of Theorem \ref{ite}]
	We assume  both $({\bf P1})_{s-1}$  (cf. \eqref{P1e}, with $s$ replacing by $s-1$) and  $({\bf P2})_{s}$ (cf.   \eqref{P2e})  hold true. We aim to prove  $({\bf P1})_{s}$ holds true as well. 	
	For convenience, we rewrite $\bm x_1=\bm x, \bm x_2=\bm y$ and denote $O_i=B_{L_s}(\bm x_i)$ for $1\le i\le 2$. 
	Then  $\|\bm x_1-\bm x_2\|> 2L_s$. We define the following events for $1\le i\le 2$ and $E\in I$:
	\begin{align*}\begin{split}
			\textbf{A}_i(E)&:\ O_i\text{ is $(\kappa_s,E)$-bad},\\
			\textbf{B}_i(E)&:\ \exists U_i\in \mathscr{T}_{s,\bm x_i}\text{ s.t.,  $U_i$ is $E$-R, where $\mathscr{T}_{s,\bm x_i}$ is defined in \eqref{Tsx}},\\
			\textbf{C}_i(E)&:\ O_i\text{ contains four pairwise disjoint $(\kappa_{s-1},E)$-bad $L_{s-1}$-cubes},\\
			\textbf{D}&:\ \exists E\in I\text{ so that both $O_1$ and $O_2$ are $(\kappa_s,E)$-bad}.
		\end{split}   
	\end{align*}
	Using the {\bf Coupling Lemma} (cf. Lemma \ref{cl})  yields
	\begin{align}
			\nonumber\mathbb{P}\left\{\textbf{D}\right\}&\le\Po\left\{\bigcup_{E\in I}(\textbf{A}_1(E)\cap\textbf{A}_2(E))\right\}\le \mathbb{P}\left\{\bigcup_{E\in I}((\textbf{B}_1(E)\cup\textbf{C}_1(E))\cap(\textbf{B}_2(E)\cup\textbf{C}_2(E)))\right\}\\
			\nonumber&\le\Po\left\{\bigcup_{E\in I}(\textbf{B}_1(E)\cap\textbf{B}_2(E))\right\}+\Po\left\{\bigcup_{E\in I}(\textbf{B}_1(E)\cap\textbf{C}_2(E))\right\}\\
			\nonumber&\ \ +\Po\left\{\bigcup_{E\in I}(\textbf{C}_1(E)\cap\textbf{B}_2(E))\right\}+\Po\left\{\bigcup_{E\in I}(\textbf{C}_1(E)\cap\textbf{C}_2(E))\right\}\\
			\label{x6}&\le \Po\left\{\bigcup_{E\in I}(\textbf{B}_1(E)\cap\textbf{B}_2(E))\right\}+3\Po\left\{\bigcup_{E\in I}\textbf{C}_1(E)\right\}.
	\end{align}
	From  the validity of \eqref{P1e} (with $s$ replacing by $s-1$), $p>5d$, $L_s=[L_{s-1}^{\alpha}]$ and $\frac{5}{4}< \alpha<\frac{2p}{2d+p}$, we obtain 
	\begin{align}\label{x7}
		\begin{split}
			\Po\left\{\bigcup_{E\in I}\textbf{C}_1(E)\right\}\le C(d)L_s^{4d}(L_{s-1}^{-2p})^2\le\frac{1}{4}L_s^{-2p}.
		\end{split}
	\end{align}
	From  the validity of  \eqref{P2e}, we have 
	\begin{align}\label{D1D2}
		\Po\left\{\bigcup_{E\in I}(\textbf{B}_1(E)\cap\textbf{B}_2(E))\right\}\le \frac{1}{4}L_s^{-2p}.
	\end{align}
	Combining \eqref{x6}, \eqref{x7} and \eqref{D1D2} implies  $\Po\{\textbf{D}\}\le L_s^{-2p}$.
	
	This concludes the proof of Theorem \ref{ite}.
\end{proof}
\subsection{Proof of Theorem \ref{thm1}}

In this subsection, we complete the proof of Theorem \ref{thm1} via combining  Theorems \ref{P1}, Theorem \ref{ite} and Theorem \ref{VP2}. 
\begin{proof}[Proof of Theorem \ref{thm1}]
Note first that $\kappa_{\infty}<\kappa_s\le \kappa_{s-1}\le \kappa_0\le \frac{\g}{5}$.  Then  it suffices to prove  $({\bf P1})_{s}$ for all $s\geq 0$. In fact, from Theorem \ref{P1}, we know $({\bf P1})_{0}$ holds true. The validity of   $({\bf P2})_{s}$  for all $s\geq1$ is guaranteed by Theorem \ref{VP2}. Then applying Theorem \ref{ite} implies $({\bf P1})_{1}$ holds true. Repeating this procedure leads to the desired proof, i.e., $({\bf P1})_{s}$ holds true  for all $s\geq 0$.



\end{proof}

\section{ Proof of theorem \ref{thm2}}\label{pot2}
In this section, we aim to prove the localization theorem (cf. Theorem \ref{thm2}) by using Theorem \ref{thm1} and the Shnol's theorem concerning generalized eigenvalues (and generalized eigenfunctions). This scheme was first introduced in  \cite{FMSS85} to prove Anderson localization for the Anderson model on $\Z^d.$

Let $\psi=\{\psi(\bm x)\}_{\bm x\in\Z^d}\in\C^{\Z^d}$ satisfy $H_{\omega}\psi=E\psi$. Assume further the Green's function  $\mathcal{G}_{B}^E$ exists for some $B\subset\Z^d$. Then for any $\bm x\in B$, we have the Poisson's identity
\begin{align}\label{PI}
	\psi(\bm x)=-\ep\sum_{\bm x'\in B,\bm x''\notin B}\mathcal{G}_{B}^E(\bm x,\bm x')\cdot\Gamma_{\phi}(\bm x',\bm x'')\cdot\psi(\bm x'').
\end{align}

We begin with the following definition.

\begin{defn}\label{yyy2}
An energy $E\in\mathbb{R}$ is called generalized eigenvalue, if there exists some $\psi\in{\C}^{{\Z}^d}$ satisfying $\psi(\bm 0)=1$, $ |\psi(\bm x)|\le (1+\|\bm x\|)^d$ and $H_{\omega}\psi=E\psi$. We call such $\psi$ the generalized eigenfunction.
\end{defn}

We  need   the following     Shnol’s Theorem, which applies to  the long-range operator.  
 
\begin{lem}[\cite{Han19}]
	Let $\mathscr{E}_{\omega}$ be the set of all generalized eigenvalues of $H_{\omega}$. Then we have $\mathscr{E}_{\omega}\subset\sigma(H_{\omega})$, $\nu_{\omega}(\sigma(H_{\omega})\setminus\mathscr{E}_{\omega})=0$, where $\nu_{\omega}$ denotes some complete spectral measure of $H_{\omega}$.
\end{lem}

In what follows,  we fix $L_0=\underline{L}_0$, $\ep_0$, $\et$ and $I=[E_0-\et,E_0+\et]$ in Theorem \ref{thm1}.

From Theorem \ref{thm1}, we have for $0<\ep<\ep_0$ and $s\ge0$,
	\begin{align}\label{2bad}
	\Po\{\exists E\in I\text{ s.t.,  both $B_{L_s}(\bm x)$ and $B_{L_s}(\bm y)$ are $(\kappa_\infty,E)$-bad}\}\le L_s^{-2p}
\end{align}
for all $\|\bm x-\bm y\|>2L_s$, where $L_{s+1}=[L_s^{\alpha}]$.

We then prove our main result on localization.

\begin{proof}[Proof of Theorem \ref{thm2}]
	For any $s\ge0$,  define $A_{s+1}=B_{10L_{s+1}}\setminus B_{3L_s}$ and the event
	\begin{align*}
		\textbf{E}_s:\ \exists E\in I, {\rm s.t., for}\ \forall \bm y\in A_{s+1},\ \text{both $B_{L_s}$ and $B_{L_s}(\bm y)$ are $(\kappa_{\infty},E)$-bad}.
	\end{align*}
	Thus from $p>5d$, $\alpha\in\left(\frac{5}{4},\frac{2p}{p+2d}\right)$ and \eqref{2bad}, it follows that 
	\begin{align*}
		\Po\{\textbf{E}_s\}&\le (20L_{s+1}+1)^d L_s^{-2p}\le C(d) L_{s}^{-2p+\alpha d},\\
		\sum_{s\ge0}\Po\{\textbf{E}_s\}&\le\sum_{s\ge0}C(d) L_{s}^{-2p+\alpha d}<\infty.
	\end{align*}
	Then by the Borel-Cantelli lemma, we have $\Po\{\textbf{E}_s\text{ occurs infinitely often}\}=0$. We define  $\Omega_0$ to be the event so that  $\textbf{E}_s$ occurs only finitely often. Then  $\Po(\Omega_0)=1$.
	
	Let $E\in I$ be a generalized eigenvalue and $\psi$ be its generalized eigenfunction. In particular $\psi(\bm 0)=1$. Suppose now there exist infinitely many $L_s$ so that all  $B_{L_s}$ are $(\kappa_{\infty},E)$-good. Then from the Poisson's indentity \eqref{PI} and \eqref{ENR}--\eqref{y6}, we obtain
	\begin{align*}
		1=|\psi(\bm 0)|&\le \sum_{\bm x'\in B_{L_s},\bm x''\notin B_{L_s}}|\mathcal{G}_{B_{L_s}}^E(\bm 0,\bm x')|\cdot e^{-\g\log^{\rho}(1+\|\bm x'-\bm x''\|)}\cdot(1+\|\bm x''\|)^d\\
		&\le (\text{I})+(\text{II}),
	\end{align*}
	where
	\begin{align*}
		(\text{I})&=\sum_{\|\bm x'\|\le L_{s}^{\frac{4}{5}},\|\bm x''\|\ge L_s}e^{\log^{\rho'}L_s}\cdot e^{-\g\log^{\rho}(1+\|\bm x'-\bm x''\|)}\cdot(1+\|\bm x''\|)^d,\\
		(\text{II})&=\sum_{L_{s}^{\frac{4}{5}}<\|\bm x'\|\le L_s ,\|\bm x''\|\ge L_s}e^{-\kappa_\infty\log^{\rho}(1+\|\bm x'\|)}\cdot e^{-\g\log^{\rho}(1+\|\bm x'-\bm x''\|)}\cdot(1+\|\bm x''\|)^d.
	\end{align*}
For (I), we get since $\|\bm x'-\bm x''\|\ge \|\bm x''\|-L_s^{\frac{4}{5}}\ge\frac{1}{2}\|\bm x''\|\ge\frac{1}{2}L_s$,
\begin{align*}
	(\text{I})&\le \sum_{\|\bm x''\|\ge L_s}(2L_s+1)^de^{\log^{\rho'}L_s}e^{-\g\log^{\rho}(1+\frac{\|\bm x''\|}{2})}(1+\|\bm x''\|)^d\\
	&\le (2L_s+1)^de^{\log^{\rho'}L_s}e^{-\frac{\g}{2}\log^{\rho}(1+\frac{L_s}{2})}\to0\ (\text{as $L_s\rightarrow\infty$}).
\end{align*}
For (II), we have by using \eqref{qua},
\begin{align*}
	(\text{II})&\le \sum_{\|\bm x''\|\ge L_s}(2L_s+1)^de^{-\kappa_\infty\log^{\rho}(1+\|\bm x''\|)+\g C(\rho)\log^{\rho}2}(1+\|\bm x''\|)^d\\
	&\le (2L_s+1)^de^{-\frac{\kappa_{\infty}}{2}\log^{\rho}(1+L_s)}\rightarrow0\ (\text{as $L_s\to\infty$}).
\end{align*}
This implies that for any generalized eigenvalue $E$, there exist only finitely many $L_s$ so that $B_{L_s}$ is $(\kappa_\infty,E)$-good.

In the following,  we fix $\omega\in\Omega_0$.

 From the above arguments, we have shown that there exists some $s_0(\omega)>0$ such that for $\forall s\ge s_0$, all $B_{L_s}(\bm x)$ with $\bm x\in A_{s+1}$ are $(\kappa_\infty,E)$-good. We define  $\tilde{A}_{s+1}=B_{8L_{s+1}}\setminus B_{2L_s}$ which satisfies  $\tilde{A}_{s+1}\subset A_{s+1}$. We will show for $s\ge s_1(\beta,d,\g,\rho,\rho',\alpha,\kappa_\infty,\omega)>0$ the following holds true:
 \begin{align}\label{decay}
 	|\psi(\bm x)|\le e^{-\frac{\kappa_{\infty}}{2\alpha^{\rho}}\log^{\rho}(1+\|\bm x\|)}\text{ for $\bm x\in\tilde{A}_{s+1}$}.
 \end{align}
 Once \eqref{decay} is established for all $s\ge s_1$, it follows from $\bigcup_{s\ge s_1}\tilde{A}_{s+1}=\{\bm x\in\Z^d:\ \|\bm x\|\ge 2L_{s_1}\}$ that 
 \begin{align*}
 	|\psi(\bm x)|\le e^{-\frac{\kappa_{\infty}}{2\alpha^{\rho}}\log^{\rho}(1+\|\bm x\|)}\ \text{for $\|\bm x\|\ge 2L_{s_1}$.}
 \end{align*}
 This then  implies that  $H_{\omega}$ exhibits localization on $I$. 
 
 We then prove \eqref{decay}. Note that $\omega\in\Omega_0$ and $\bm x\in\tilde{A}_{s+1}\subset A_{s+1}$. We know that $B_{L_s}(\bm x)\subset A_{s+1}$ is $(\kappa_\infty,E)$-good (cf. \eqref{y6}). Applying  \eqref{PI} again gives 
 \begin{align*}
 	|\psi(\bm x)|&\le \sum_{\bm x'\in B_{L_s}(\bm x),\bm x''\notin B_{L_s}(\bm x)}|\mathcal{G}_{B_{L_s}}^E(\bm x,\bm x')|\cdot e^{-\g\log^{\rho}(1+\|\bm x'-\bm x''\|)}\cdot(1+\|\bm x''\|)^d\\
 	&\le (\text{III})+(\text{IV}),
 \end{align*}
 where
 	\begin{align*}
 	(\text{III})&=\sum_{\|\bm x-\bm x'\|\le L_{s}^{\frac{4}{5}},\|\bm x-\bm x''\|\ge L_s}e^{\log^{\rho'}L_s}\cdot e^{-\g\log^{\rho}(1+\|\bm x'-\bm x''\|)}\cdot(1+\|\bm x''\|)^d,\\
 	(\text{IV})&=\sum_{L_{s}^{\frac{4}{5}}<\|\bm x-\bm x'\|\le L_s ,\|\bm x-\bm x''\|\ge L_s}e^{-\kappa_\infty\log^{\rho}(1+\|\bm x-\bm x'\|)}\cdot e^{-\g\log^{\rho}(1+\|\bm x'-\bm x''\|)}\cdot(1+\|\bm x''\|)^d.
 \end{align*}
 For (III),  from  $\|\bm x'-\bm x''\|\ge \|\bm x-\bm x''\|-L_s^{\frac{4}{5}}\ge\frac{1}{2}\|\bm x-\bm x''\|\ge\frac{1}{2}L_s$, $ 2L_s\le \|\bm x\|\le 8L_{s+1}$ and $(1+\|\bm x''\|)^d\le (1+\|\bm x\|)^d(1+\|\bm x-\bm x''\|)^d$, we have
 \begin{align*}
 (\text{III})&\le \sum_{\|\bm x-\bm x''\|\ge L_s}(2L_s+1)^de^{\log^{\rho'}L_s}e^{-\g\log^{\rho}(1+\frac{\|\bm x-\bm x''\|}{2})}(1+\|\bm x''\|)^d\\
 &\le (2L_s+1)^de^{\log^{\rho'}L_s}e^{-\frac{\g}{2}\log^{\rho}(1+\frac{L_s}{2})}(1+\|\bm x\|)^d\\
 &\le e^{-\frac{\g}{4}\log^{\rho}(1+\frac{L_s}{2})}\le \frac{1}{2}e^{-\frac{\kappa_{\infty}}{2\alpha^{\rho}}\log^{\rho}(1+\|\bm x\|)}.
 \end{align*}
 For (IV), we have by \eqref{qua}, $2L_s\le \|\bm x\|\le 8L_{s+1}$ and $(1+\|\bm x''\|)^d\le (1+\|\bm x\|)^d(1+\|\bm x-\bm x''\|)^d$,
 \begin{align*}
 (\text{IV})&\le \sum_{\|\bm x-\bm x''\|\ge L_s}(2L_s+1)^de^{-\kappa_\infty\log^{\rho}(1+\|\bm x-\bm x''\|)+\g C(\rho)\log^{\rho}2}(1+\|\bm x''\|)^d\\
 &\le (2L_s+1)^de^{-\frac{3\kappa_{\infty}}{4}\log^{\rho}(1+L_s)}(1+\|\bm x\|)^d\\
 &\le \frac{1}{2}e^{-\frac{\kappa_{\infty}}{2\alpha^{\rho}}\log^{\rho}(1+\|\bm x\|)}.
 \end{align*}
 Combining the above estimates implies  for $\forall \bm x\in\tilde{A}_{s+1}$,
 \begin{align*}
 	|\psi(\bm x)|\le e^{-\frac{\kappa_{\infty}}{2\alpha^{\rho}}\log^{\rho}(1+\|\bm x\|)}.
 \end{align*}
 
 We complete the proof. 
\end{proof}


\appendix{}
\section{}\label{APPqua}
\begin{proof}[Proof of Lemma \ref{qua}]
Let 
	\begin{align*}
	F(x_1,x_2,\cdots,x_n)=\log^{\rho}(1+\sum\limits_{i=1}^n x_i)-\sum\limits_{i=1}^n\log^{\rho}(1+x_i),\ (x_1,x_2,\cdots, x_n)\in[0,+\infty)^n.
\end{align*}
	Direct computation shows for $1\leq i\leq n,$
	\begin{align}
		\nonumber \partial_{x_i}F(x_1,\cdots, x_{i},\cdots,x_n)&=\rho\left(\frac{\log^{\rho-1}(1+\sum\limits_{i=1}^n x_i)}{1+\sum\limits_{i=1}^n x_i}-\frac{\log^{\rho-1}(1+x_i)}{1+x_i}\right)\\
		\label{Fh}&=\rho(h(1+\sum\limits_{i=1}^n x_i)-h(1+x_i)),
	\end{align}
	where  $h(s)=\dfrac{\log^{\rho-1}(s)}{s}$ ($s\geq 1$).  
	Then 
	\begin{align*}
		\frac{dh}{ds}=\frac{\log^{\rho-2}(s)\left(\rho-1-\log s\right)}{s^2},
	\end{align*}
which  implies that  $h(s)$ is  increasing  in $(1,e^{\rho-1})$ and decreasing  in $(e^{\rho-1},+\infty)$.  Thus this combined with  \eqref{Fh} deduces that if   $x_i\geq e^{\rho-1}-1$ for some $1\leq i\leq n$,  then 
	\begin{align}\label{Fx}
		 \partial_{x_i}F(x_1,\cdots,x_i,\cdots,x_n)\leq 0. 
	\end{align}
	
	Next, we assert that
	\begin{align*}
		\sup\limits_{(x_1,\cdots,x_n)\in [0,+\infty)^n}F(x_1,\cdots,x_n)=\sup\limits_{(x_1,\cdots,x_n)\in [0,e^{\rho-1}]^n}F(x_1,\cdots,x_n).
	\end{align*}
	For $x\in [0,+\infty)$, define  $x_*=\min\{x,e^{\rho-1}\}$. If $x_j>e^{\rho-1}$ for some $1\le j\le n$, then  combining  the mean value theorem  and \eqref{Fx} yields that  there  exists  some $\xi\in (e^{\rho-1},x_j)$ such that
	\begin{align*}
			&\ \  F(x_1,\cdots,x_{j-1},x_j,x_{j+1},\cdots,x_n)-F(x_1,\cdots,x_{j-1},e^{\rho-1},x_{j+1},\cdots,x_n)\\
			=&\ \ \partial_{x_j}F(x_1,\cdots,x_{j-1},\xi,x_{j+1},\cdots,x_n)(x_j-e^{\rho-1})\le0.
 \end{align*}
	Therefore
	\begin{align*}
	F(x_1,\cdots,x_{j-1},x_j,x_{j+1},\cdots,x_n)\le F(x_1,\cdots,x_{j-1},(x_j)_*,x_{j+1},\cdots,x_n),
\end{align*}
	which implies 
	\begin{align*}
	\sup\limits_{(x_1,\cdots,x_n)\notin [0,e^{\rho-1}]^n}F(x_1,\cdots,x_n)\le \sup\limits_{(x_1,\cdots,x_n)\in [0,e^{\rho-1}]^n}F(x_1,\cdots,x_n)
\end{align*}
and
	\begin{align*}
	\sup\limits_{(x_1,\cdots,x_n)\in [0,+\infty)^n}F(x_1,\cdots,x_n)=\sup\limits_{(x_1,\cdots,x_n)\in [0,e^{\rho-1}]^n}F(x_1,\cdots,x_n).
\end{align*}

Finally,  since  $F$  is continuous on $[0,e^{\rho-1}]^n$, there  exists some $(x_1^*,\cdots,x_n^*)\in[0,e^{\rho-1}]^n$ such that 
	\begin{align*}
		F(x_1^*,\cdots,x_n^*)=\sup\limits_{(x_1,\cdots,x_n)\in [0,e^{\rho-1}]^n}F(x_1,\cdots,x_n).
	\end{align*}
 Define $r=\#\{1\le i\le n:\ x_i^*\ne0\}$.
If $r\le 1$, we have
\begin{align*}
 F(x_1^*,\cdots,x_n^*)=0\le C(\rho)\log^{\rho} n,
\end{align*}
which implies \eqref{quaeq}. 
If $r\ge 2$,  without loss of generality, we can  assume that $x_i^{*}\in(0,e^{\rho-1}]$ for $1\le i\le r$ and $x_{j}^*=0$ for $r+1\le j\le n$.  Define $G(x_1,\cdots, x_r)=F(x_1,\cdots,x_r,0,\cdots,0)$. We know  $(x_1^*,\cdots,x_r^*)\in (0, +\infty)^r$ is the maximum point of $G(x_1,\cdots,x_r)$ on $[0,+\infty)^r$. Hence  we get 
	\begin{align}\label{vv2}
	\partial_{x_i}F(x_1^*,\cdots,x_r^*,0,\cdots,0)=\partial_{x_i}G(x_1^*,\cdots,x_r^*)=0,\  1\le i\le r.
\end{align}
If in addition there exist some $1\le i\ne j\le r$ such that $x_i^*\ne x_j^*$, then  it follows from (\ref{vv2}) that 
\begin{align*}
	h(1+\sum\limits_{i=1}^n x_i^*)=h(1+x_i^*)=h(1+x_j^*),
\end{align*}
	 which means the equation $h(s)=a$ has  three distinct roots in $(1,+\infty)$ for some $a$. This contradicts with the fact that   $h(s)$ is  increasing  in $(1,e^{\rho-1})$ and decreasing  in $(e^{\rho-1},+\infty)$. Therefore, we must have $x_1^*=x_2^*=\cdots=x_r^*\in(0,e^{\rho-1}]$. Moreover, we obtain
	\begin{align*}
F(x_1^*,\cdots,x_r^*,0,\cdots,0)&=\log^{\rho}(rx_1^*+1)-r\log^{\rho}(x_1^*+1)\\
			&\le \log^{\rho}(ne^{\rho-1}+1)\le C(\rho)\log^{\rho}n,
	\end{align*}
which also implies \eqref{quaeq}.  We complete the proof of Lemma \ref{qua}. 
\end{proof}

\section{}\label{APPdan}
\begin{proof}[Proof of Lemma \ref{neighbor}]
	Let $\bm x=(x_1,\cdots,x_d)\in\Z^d$ and $\bm z=(z_1,\cdots,z_d)\in B_{L_s}(\bm x)$. For $1\le k\le d$, we define
	\begin{align}\label{hatz}
		\hat{z}_k=\left\{\begin{array}{ll}
			z_{k}, &	|x_{k}-z_{k}|\le L_{s}-L_{s-1},\\
			x_{k}+L_{s}-L_{s-1}, & z_k-x_{k}>L_{s}-L_{s-1},\\
			x_{k}-L_{s}+L_{s-1}, & x_{k}-z_{k}>L_{s}-L_{s-1}.
		\end{array}\right.
	\end{align}
	Let $\hat{\bm z}=(\hat{z}_1,\cdots,\hat{z}_d)$. Then $\|\bm z-\hat{\bm z}\|\le L_{s-1}$ and $\|\hat{\bm z}-\bm x\|\le L_{s}-L_{s-1}$, which implies
	\begin{align*}
		\bm z\in B_{L_{s-1}}(\hat{\bm z})\subset B_{L_s}(\bm x).
	\end{align*}

	If $\bm y=(y_1,\cdots,y_d)\in B_{L_s}(\bm x)\setminus B_{L_{s-1}}(\hat{\bm z})$, there is some $1\le k'\le d$ such that $|y_{k'}-\hat{z}_{k'}|=\|\bm y-\hat{\bm z}\|\ge L_{s-1}+1$. We then prove $\|\bm y-\bm z\|\ge \|\bm y-\hat{\bm z}\|\ge L_{s-1}+1$. We divide the discussion into the following cases:
			\begin{itemize}
		\item[\textbf{Case 1}:] $|z_{k'}-x_{k'}|\le L_s-L_{s-1}$. From \eqref{hatz} and $\bm y=(y_1,\cdots,y_d)\in B_{L_s}(\bm x)\setminus B_{L_{s-1}}(\hat{\bm z})$, we have
		\begin{align*}
			|y_{k'}- z_{k'}|=|y_{k'}-\hat{z}_{k'}|\ge L_{s-1}+1.
		\end{align*}
	
		\item[\textbf{Case 2}:] $z_{k'}-x_{k'}>L_{s}-L_{s-1}$. By \eqref{hatz} and $\bm y=(y_1,\cdots,y_d)\in B_{L_s}(\bm x)\setminus B_{L_{s-1}}(\hat{\bm z})$,  we also have 
		\begin{align*}
		|y_{k'}-x_{k'}-L_{s}+L_{s-1}|=|y_{k'}-\hat{z}_{k'}|\ge L_{s-1}+1.
	\end{align*}
		 Moreover, if $y_{k'}-x_{k'}-L_{s}+L_{s-1}\ge L_{s-1}+1$,  then 
		 \begin{align*}
		 y_{k'}-x_{k'}\ge L_{s}+1,
		\end{align*}
		  which contradicts with $\bm y\in B_{L_s}(\bm x)$. Hence
		  \begin{align*}
		  y_{k'}-z_{k'}<y_{k'}-x_{k'}-L_{s}+L_{s-1}=y_{k'}-\hat{z}_{k'}\le -(L_{s-1}+1)
		\end{align*}
		   and
		  \begin{align*}
		  |y_{k'}-z_{k'}|> |y_{k'}-\hat{z}_{k'}|\ge L_{s-1}+1.
		  \end{align*}
		\item[\textbf{Case 3}:] $x_{k'}-z_{k'}>L_{s}-L_{s-1}$.  Similar to the analysis in the  {\bf Case 2}, we  get 
		 \begin{align*}
			|y_{k'}-z_{k'}|> |y_{k'}-\hat{z}_{k'}|\ge L_{s-1}+1.
		\end{align*}
	\end{itemize}
		Therefore, $\|\bm y-\bm z\|\ge |y_{k'}- z_{k'}|=\|\bm y-\hat{\bm z}\|\ge L_{s-1}+1$.
		
		Now, if $\bm y=(y_1,\cdots,y_d)\in B_{L_s}(\bm x)\cap B_{L_{s-1}}(\bm z)$, then $|y_{k}-z_{k}|\le L_{s-1}$ for all $1\le k\le d$. We will show that $\|\bm y-\hat{\bm z}\|\le L_{s-1}$. We also  have the following  cases for $1\le k\le d$:
					\begin{itemize}
			\item [\textbf{Case 1}:] $|z_{k}-x_{k}|\le L_{s}-L_{s-1}$. Since \eqref{hatz} and $|y_{k}-z_{k}|\le L_{s-1}$, we obtain
			\begin{align*}
				|y_{k}- z_{k}|=|y_{k}-\hat{z}_{k}|\le L_{s-1}.
			\end{align*} 
			\item [\textbf{Case 2}:] $z_{k}-x_{k}>L_{s}-L_{s-1}$. By \eqref{hatz} and $|y_{k}-z_{k}|\le L_{s-1}$, we have 
			\begin{align*}
				z_k-L_{s-1}\le y_k\le x_{k}+L_{s}.
			\end{align*}
			Hence
			\begin{align*}
				-L_{s-1}<z_k-x_{k}-L_{s}\le y_k-\hat{z}_k=y_k-x_{k}-L_{s}+L_{s-1}\le L_{s-1},
			\end{align*}
			which implies $|y_k-\hat{z}_k|\le L_{s-1}$.
			\item [\textbf{Case 3}:] $x_{k}-z_{k}>L_s-L_{s-1}$. From a similar argument in the above case, we get 
			\begin{align*}
				|y_{k}-\hat{z}_{k}|\le L_{s-1}.
			\end{align*}
		\end{itemize}
		Thus $\|\bm y-\hat{\bm z}\|=\max\limits_{1\le k\le d}|y_k-\hat{z}_k|\le L_{s-1}$.

This finishes the proof of Lemma \ref{neighbor}. 
\end{proof}

\begin{proof}[Proof of Lemma \ref{3kuai}]
	For $1\le i\le 3$, the construction of $\bm w_i^{*}$ is similar to that of \eqref{hx}, so we omit the detals. 

If $\bm z\in B_{L_s}(\bm x)\setminus\cup_{i=1}^{3} B_{10L_{s-1}}(\bm w_i^*)$, by \eqref{dhx3}, we have
\begin{align}\label{901}
	\|\bm z-\bm w_i\|\ge10 L_{s-1}+1,\ 1\le i\le 3.
\end{align}
From \eqref{hx} and \eqref{901}, we can get
\begin{align*}
	\|\hat{\bm z}-\bm w_i\|\ge\|\bm z-\bm w_i\|-\|\bm z-\hat{\bm z}\|\ge9L_{s-1}+1,\ 1\le i\le 3,
\end{align*}
	which implies $B_{L_{s-1}}(\hat{\bm z})\cap B_{L_{s-1}}(\bm w_i)=\emptyset$, $1\le i\le 3$. Since  there are at most three pairwise disjoint $(\kappa_{s-1},E)$-bad $L_{s-1}$-cubes in $B_{L_s}(\bm x)$, we have that  $B_{L_{s-1}}(\hat{\bm z})$ is $(\kappa_{s-1},E)$-good.
	
	This proves Lemma \ref{3kuai}. 
\end{proof}
\begin{proof}[Proof of Lemma \ref{dancub}]
	We  divide the discussion into the following  cases: 
	\begin{itemize}
		\item [\textbf{Case 1}:] $\dist(B_{10L_{s-1}}(\bm w_{i}^{*}),B_{10L_{s-1}}(\bm w_{j}^{*}))\ge 10L_{s-1}$. For $1\le i\ne j\le 3$, we define  
		$$B_{i}=B_{10L_{s-1}}(\bm w_{i}^{*})$$ with $l_i=\diam(B_i)=20L_{s-1}, 1\le i\le 3.$
		\item [\textbf{Case 2}:] 
		 Without loss of generality, we can assume
		 \begin{align*}
		 \dist(B_{10L_{s-1}}(\bm w_{1}^{*}),B_{10L_{s-1}}(\bm w_{2}^{*}))< 10L_{s-1}.
		\end{align*}
		  In this case, there are some $\bm y_1\in B_{10L_{s-1}}(\bm w_{1}^{*})$ and $\bm y_2\in B_{10L_{s-1}}(\bm w_{2}^{*})$ such that $\|\bm y_1-\bm y_2\|<10L_{s-1}$.  Then 
		\begin{align}\label{902}
			\|\bm w_1^{*}-\bm w_2^{*}\|\le \|\bm w_1^*-\bm y_1\|+\|\bm y_1-\bm y_2\|+\|\bm y_2-\bm w_2^*\|<30L_{s-1}.
		\end{align}
		 Let $\bm w_{i}^{*}=(w_{i,1},\cdots,w_{i,d})$, $1\le i\le3$. We define for $1\le k\le d$,
		\begin{align}\label{903}
			v_{1,k}=\left\{\begin{array}{ll}
				\frac{w_{1,k}+w_{2,k}}{2}, & \text{if $w_{1,k}+w_{2,k}\equiv0\ ({\rm mod}\ 2)$},\\
				\frac{w_{1,k}+w_{2,k}+1}{2}, & \text{if $w_{1,k}+w_{2,k}\equiv1\ ({\rm mod}\ 2)$},
			\end{array}\right.
		\end{align}
		and $\bm v_{1}=(v_{1,k},\cdots,v_{1,d})$. Then by \eqref{902} and \eqref{903}, we get
		\begin{align}\label{904}
			\|\bm v_1-\bm w_{1}^{*}\|\le15L_{s-1}\text{ and }\|\bm v_1-\bm w_{2}^{*}\|\le15L_{s-1}.
		\end{align}
		 Define again for $1\le k\le d$,
			\begin{align}\label{905}
			v_{2,k}=\left\{\begin{array}{ll}
				v_{1,k}, &	|v_{1,k}-x_{k}|\le L_s-40L_{s-1},\\
				x_k+L_s-40L_{s-1}, & v_{1,k}-x_k>L_s-40L_{s-1},\\
				x_k-L_s+40L_{s-1}, & x_k-v_{1,k}>L_s-40L_{s-1},
			\end{array}\right.
		\end{align}
		and denote  $\bm v_{2}=(v_{2,1},\cdots,v_{2,d})$.  Then $\|\bm v_{2}-\bm x\|\le L_s-40L_{s-1}$. Next, we will show $\|\bm v_{2}-\bm w_1^{*}\|\le 30L_{s-1}$ and $\|\bm v_{2}-\bm w_2^{*}\|\le 30L_{s-1}$. For this purpose,  we have the following cases  for $1\le k\le d$. 
		\begin{itemize}
			\item [\textbf{Case i}:] $|v_{1,k}-x_{k}|\le L_s-40L_{s-1}$. From \eqref{904} and \eqref{905}, we obtain
			\begin{align*}
				|v_{2,k}-w_{1,k}|=|v_{1,k}-w_{1,k}|\le15L_{s-1}
			\end{align*}
			and 
			\begin{align*}
				|v_{2,k}- w_{2,k}|=|v_{1,k}-w_{2,k}|\le15L_{s-1}.
			\end{align*}
			\item [\textbf{Case ii}:] $v_{1,k}-x_k>L_s-40L_{s-1}$. Since  \eqref{903} and $v_{1,k}-x_k>L_s-40L_{s-1}$, there exists some $a\in \{w_{1,k},w_{2,k}\}$ such that
			\begin{align*}
				L_s-10L_{s-1}\ge a-x_k\ge L_s-40L_{s-1}.
			\end{align*}
			Without loss of generality, we assume that 
			\begin{align*}
			L_s-10L_{s-1}\ge w_{1,k}-x_k\ge L_s-40L_{s-1}.
		\end{align*}
			 Since $|w_{1,k}-w_{2,k}|\le \|\bm w_{1}^{*}-\bm w_{2}^{*}\|<30L_{s-1}$ (cf. \eqref{902}), we also have 
			 \begin{align*}
			 	L_s-10L_{s-1}\ge w_{2,k}-x_k\ge L_s-70L_{s-1}.
			 \end{align*}
			  Therefore,
			\begin{align*}
				|v_{2,k}-w_{1,k}|\le 30L_{s-1}\ \text{and}\ |v_{2,k}-w_{2,k}|\le 30L_{s-1}.
			\end{align*}
			\item [\textbf{Case iii}:] $x_k-v_{1,k}>L_{s}-40L_{s-1}$. By a similar argument  as in  \textbf{Case ii}, we can get 
			\begin{align*}
				|v_{2,k}-w_{1,k}|\le 30L_{s-1}\ \text{and}\ |v_{2,k}-w_{2,k}|\le 30L_{s-1}.
			\end{align*}
		\end{itemize}
		Combining the estimates in the above three cases leads to  
		\begin{align*}
			\|\bm v_2-\bm w_1^{*}\|=\max_{1\le k\le d}|v_{2,k}-w_{1,k}|\le 30L_{s-1}
		\end{align*}
		and
		\begin{align*}
				\|\bm v_2-\bm w_2^{*}\|=\max_{1\le k\le d}|v_{2,k}-w_{2,k}|\le 30L_{s-1}.
		\end{align*}
		Therefore,
		\begin{align*}
			B_{10L_{s-1}}(\bm w_1^{*})\cup B_{10L_{s-1}}(\bm w_2^{*})\subset B_{40L_{s-1}}(\bm v_2)\subset B_{L_s}(\bm x).
		\end{align*}
	Similar to the construction of $B_{40L_{s-1}}(\bm v_2)$ (cf. \eqref{905}), there is a $\bm v_3$ such that 
	\begin{align*}
			B_{10L_{s-1}}(\bm w_3^{*})\subset B_{40L_{s-1}}(\bm v_3)\subset B_{L_s}(\bm x).
		\end{align*}
		To finish the construction of $B_i$, it remains to deal with the following sub-cases  of {\bf Case 2}. 
		\item [\textbf{Case 2-1}:] $\dist(B_{40L_{s-1}}(\bm v_2),B_{40L_{s-1}}(\bm v_3))\ge 10L_{s-1}$. In this sub-case, we define 
		$$B_{1}=B_{40L_{s-1}}(\bm v_2),\ B_{2}=B_{10L_{s-1}}(\bm w_3^{*}),\ B_3=\emptyset$$
		 with $l_1=\diam(B_1)=80L_{s-1}$, $l_2=\diam(B_2)=4L_{s-1}$, $l_3=0$.
		\item [\textbf{Case 2-2}:] $\dist(B_{40L_{s-1}}(\bm v_2),B_{40L_{s-1}}(\bm v_3))<10L_{s-1}$. In  this sub-case,  similar to the  construction of $B_{40L_{s-1}}(\bm v_2)$ (cf. \eqref{905}), we can obtain that there is some $\bm v_4$ such that
		\begin{align*}
			B_{40L_{s-1}}(\bm v_2)\cup B_{40L_{s-1}}(\bm v_3)\subset B_{130L_{s-1}}(\bm v_4)\subset B_{L_s}(\bm x).
		\end{align*}
		 Then we define  
		$$B_{1}=B_{130L_{s-1}}(\bm v_4),\ B_{2}=B_3=\emptyset $$
		with $l_1=\diam(B_1)=260L_{s-1}$, $l_2=l_3=0$.
	\end{itemize}

Finally, from the assumption of Lemma \ref{cl} (2),  it follows  that $B_i$ is $E$-NR, $1\le i\le 3$. Since $ B_{L_s}(\bm x)\setminus\Pi\subset B_{L_s}(\bm x)\setminus\bigcup_{i=1}^{3}B_{10L_{s-1}}(\bm w_i^*)$ and Lemma \ref{3kuai}, we have  that $B_{L_{s-1}}(\hat{\bm z})$ is $(\kappa_{s-1},E)$-good for $\bm z\in B_{L_s}(\bm x)\setminus \Pi$.

This completes the proof of Lemma \ref{dancub}. 
\end{proof}
\section{}\label{CN}
\begin{proof}[Proof of \eqref{countn}]
 We have the following cases when $0\le k\le n-1$.
\begin{itemize}
	\item[\textbf{Case 1}:] $\bm z_{k}\in B_i$ for some $1\le i\le 3$ and $\|\bm z_{k+1}-\tilde{\bm z}_k\|\le L_{s-1}$. In this case, we obtain
\begin{align*}
	\dist(\bm z_{k+1},B_i)\le \|\bm z_{k+1}-\tilde{\bm z}_k\|\le L_{s-1},
\end{align*}
since $\tilde{\bm z}_k\in B_i$. From Lemma \ref{dancub} (1) and $\bm z_{k+1}\in B_{L_s}(\bm x)\setminus B_i$, we have
\begin{align*}
	\bm z_{k+1}\in B_{L_s}(\bm x)\setminus \Pi.
\end{align*}
Therefore,  applying Lemma \ref{dancub} (5) implies  $B_{L_{s-1}}(\hat{\bm z}_{k+1})$ is $(\kappa_{s-1},E)$-good. Next,  from  Lemma \ref{gfeg},   $\tilde{\bm z}_{k+1}=\bm z_{k+1}$, $\bm z_{k+2}\in B_{L_s}(\bm x)\setminus B_{L_{s-1}}(\hat{\bm z}_{k+1})$ and \eqref{dhx1}, we get
\begin{align*}
\|\bm z_{k+2}-\tilde{\bm z}_{k+1}\|=\|\bm z_{k+2}-\bm z_{k+1}\|\ge\|\bm z_{k+2}-\hat{\bm z}_{k+1}\|\ge L_{s-1}+1.
\end{align*}
So $k+1\in \mathscr{N}$.
\item[\textbf{Case 2}:] $\bm z_{k}\in B_i$ for some $1\le i\le 3$ and $\|\bm z_{k+1}-\tilde{\bm z}_k\|\ge L_{s-1}+1$. It is obvious that $k\in \mathscr{N}$.
\item[\textbf{Case 3:}] $\bm z_{k}\in B_{L_s}(\bm x)\setminus\Pi$. By Lemma \ref{dancub}, we know  $B_{L_{s-1}}(\hat{\bm z}_{k})$ is $(\kappa_{s-1},E)$-good. From Lemma \ref{gfeg}, $\tilde{\bm z}_{k}=\bm z_{k}$, $\bm z_{k+1}\in B_{L_s}(\bm x)\setminus B_{L_{s-1}}(\hat{\bm z}_{k})$ and \eqref{dhx1}, it follows that 
\begin{align*}
\|\bm z_{k+1}-\tilde{\bm z}_{k}\|=\|\bm z_{k+1}-\bm z_{k}\|\ge\|\bm z_{k+1}-\hat{\bm z}_{k}\|\ge L_{s-1}+1.
\end{align*}
Thus  $k\in \mathscr{N}$.
\end{itemize}
Combining the above cases  shows that   for every $0\leq k\leq n-1$,  either $k\in \mathscr{N}$ or $k+1\in \mathscr{N}$. This implies    $\# \mathscr{N}\ge\left[\frac{n}{2}\right]$.

\end{proof}

	\section*{Acknowledgments}
The authors would like to thank the editor and two referees for their helpful suggestions. Shi  was supported by the National Key R\&D Program of China (2021YFA1001600) and NSFC (12271380).
\section*{Data Availability}
The manuscript has no associated data.
\section*{Declarations}
{\bf Conflicts of interest} \ The authors  state  that there is no conflict of interest.

	\bibliographystyle{alpha}

\end{document}